\documentclass[11pt,draft,onecolumn]{IEEEtran}

\usepackage{graphicx}
\usepackage{epsfig}
\usepackage{amsmath}
\usepackage{amssymb}
\usepackage{subfigure}
\usepackage{wrapfig}
\usepackage{times,color}
\usepackage{cite,footnote,xspace,syntonly,theorem,algorithm,algorithmic}
\input{mysymbol.sty}
\newcommand{\calN}{{\cal N}}

\newcommand{\calJ}{{\cal J}}
\newcommand{\calbbs}{\boldsymbol{{\mathit{s}}}}%{\boldsymbol{{\mathit{s}}}}

\def \bbOmega    {{{\bf{\Omega}}}}

\newcommand{\calbbz}{\boldsymbol{{\mathit{z}}}}%{\boldsymbol{{\mathit{s}}}}
\newcommand{\calbbv}{\boldsymbol{{\mathit{v}}}}%{\boldsymbol{{\mathit{s}}}}
\newcommand{\calbbu}{\boldsymbol{{\mathit{u}}}}%{\boldsymbol{{\mathit{s}}}}

\def \calN {{\cal N}}

\def \calJ {{\cal J}}
\def \calbbs {\boldsymbol{{\mathit{s}}}}%{\boldsymbol{{\mathit{s}}}}
\def \bbOmega    {{{\bf{\Omega}}}}

\newcommand{\boxedeqnmulti}[1]{%
  \[\fbox{%
      \addtolength{\linewidth}{-2\fboxsep}%
      \addtolength{\linewidth}{-2\fboxrule}%
      \begin{minipage}{\linewidth}%
      \begin{align}#1\end{align}%
      \nonumber\end{minipage}%
    }\]%
}
\long\def\symbolfootnote[#1]#2{\begingroup
\def\thefootnote{\fnsymbol{footnote}}
\footnote[#1]{#2}\endgroup} \psfull

\begin{document}
% % % % % % % % % % % % % % % % % % % % % % % % % % % % % % % % % % % % % % % %
%                         Cover Page                                          %
% % % % % % % % % % % % % % % % % % % % % % % % % % % % % % % % % % % % % % % %

\title{\huge Distributed Recursive Least-Squares:\\ 
Stability and Performance Analysis$^\dag$}

\author{\it Gonzalo Mateos, Student Member, IEEE, 
and Georgios~B.~Giannakis, Fellow, IEEE$^\ast$}

\markboth{IEEE TRANSACTIONS ON SIGNAL PROCESSING (SUBMITTED)} \maketitle
\maketitle \symbolfootnote[0]{$\dag$ Work in this
paper was supported by the NSF grants CCF-0830480 and
ECCS-0824007.} \symbolfootnote[0]{$\ast$ The authors are with the
Dept. of Electrical and Computer Engineering, University of Minnesota, 200
Union Street SE, Minneapolis, MN 55455. Tel/fax: (612)626-7781/625-2002;
Emails: \texttt{\{mate0058,georgios\}@umn.edu}}

\vspace*{-80pt}
\begin{center}
\small{\bf Submitted: }\today\\
\end{center}
\vspace*{10pt}

% % % % % % % % % % % % % % % % % % % % % % % % % % % % % % % % % % % % % % % %
%                         Abstract                                            %
% % % % % % % % % % % % % % % % % % % % % % % % % % % % % % % % % % % % % % % %

\thispagestyle{empty}\addtocounter{page}{-1}
\begin{abstract}
The recursive least-squares (RLS) algorithm has well-documented merits
for reducing complexity and storage requirements, when it comes to online 
estimation of stationary signals as well as for tracking slowly-varying 
nonstationary processes. 
In this paper, a \textit{distributed} recursive least-squares (D-RLS) 
algorithm is developed for cooperative estimation using ad hoc wireless sensor
 networks.
Distributed 
iterations are obtained by minimizing a separable reformulation of the 
exponentially-weighted least-squares cost, using the alternating-minimization 
algorithm.  Sensors carry
 out reduced-complexity tasks locally, and exchange messages with one-hop 
neighbors to consent on the network-wide estimates adaptively. A steady-state 
mean-square 
error (MSE) 
performance analysis of D-RLS is conducted, by studying a stochastically-driven
 `averaged' system that approximates the D-RLS dynamics asymptotically in 
time. For sensor 
observations that are linearly
related to the time-invariant parameter vector sought, 
the simplifying independence setting
assumptions facilitate deriving accurate
closed-form expressions for the MSE steady-state values. The problems
of mean- and MSE-sense stability of D-RLS are also investigated, and 
easily-checkable sufficient conditions are derived under which a steady-state 
is attained. Without resorting to diminishing step-sizes which compromise 
the tracking ability of D-RLS, stability ensures that  per sensor estimates 
hover inside a ball of finite radius centered at the true 
parameter vector, with high-probability, even when 
inter-sensor communication links are noisy.  Interestingly, computer
simulations demonstrate that the theoretical findings are accurate also in 
the pragmatic settings whereby sensors acquire temporally-correlated data.
\end{abstract}


%\newpage
\begin{center} \bfseries EDICS Category: SEN-DIST, SPC-PERF, SEN-ASAL 
\end{center}
%
% for peerreview papers, inserts a page break and creates the second title.
% Will be ignored for other modes.
%\IEEEpeerreviewmaketitle
%\newpage

% % % % % % % % % % % % % % % % % % % % % % % % % % % % % % % % % % % % % % % %
%                         Section I                                           %
% % % % % % % % % % % % % % % % % % % % % % % % % % % % % % % % % % % % % % % %

\newpage
\section{Introduction}
\label{sec:intro} Wireless sensor networks (WSNs), whereby large numbers of 
inexpensive sensors with constrained resources cooperate to achieve a common 
goal, constitute a promising technology for applications as diverse and crucial
as environmental monitoring, process control and fault diagnosis for the industry, 
and protection of critical infrastructure including the smart
 grid, just to name a few. Emergent WSNs have created renewed interest also in the 
field of distributed computing, calling for collaborative solutions that enable
 low-cost estimation of stationary signals as well as reduced-complexity 
tracking of nonstationary processes; see e.g.,~\cite{Xiao_Ale_Luo_GG_SPMag,Ale_Yannis_Stergios_GG_ControlsMag}.

In this paper, a \textit{distributed} recursive least-squares (D-RLS) 
algorithm is developed for estimation and tracking using ad hoc WSNs with noisy
 links, and analyzed in terms of its stability and mean-square error (MSE) 
steady-state performance. Ad hoc WSNs lack a central processing unit, and 
accordingly D-RLS performs in-network processing of the (spatially) distributed
 sensor observations. In words, a two-step iterative process takes place 
towards consenting on 
the desired global exponentially-weighted least-squares estimator (EWLSE): 
sensors perform simple local tasks to refine their current estimates, and 
exchange messages with one-hop neighbors over noisy communication channels. New
 sensor data acquired in real time enrich the estimation process and 
learn the unknown statistics `on-the-fly'. In addition, the exponential 
weighting effected 
through a forgetting factor endows D-RLS with tracking capabilities. 
This is desirable in a constantly changing environment, 
within which WSNs are envisioned to operate.

% % % % % % % % % % % % % % % % % % % % % % % % % % % % % % % % % % % % % % % %
%                         Subsection I-A                                      %
% % % % % % % % % % % % % % % % % % % % % % % % % % % % % % % % % % % % % % % %

\subsection{Prior art on distributed adaptive estimation}
\label{ssec:prior_art}
Unique challenges arising with WSNs dictate that often times sensors need to 
perform estimation in a constantly changing environment without having 
available a (statistical) model for the underlying processes of interest. This 
has motivated the development of \textit{distributed adaptive} estimation
schemes, generalizing the notion of adaptive filtering to a setup involving 
networked sensing/processing 
devices~\cite[SectionI-B]{Cattivelli_Lopes_Sayed_TSP_Diffusion_RLS}.

The incremental (I-) RLS algorithm
in~\cite{Lopes_Sayed_Incremental_RLS_2006} is one of the first such 
approaches, which sequentially incorporates new sensor data while performing
least-squares estimation. If one can afford maintaining a 
so-termed Hamiltonian cyclic path across sensors, then I-RLS yields the centralized 
EWLS benchmark estimate. Reducing the communication cost  
at a modest price in terms of estimation performance, an I-RLS variant
was also put forth in~\cite{Lopes_Sayed_Incremental_RLS_2006}; 
but the NP-hard challenge of determining a Hamiltonian cycle in 
large-size WSNs remains~\cite{Papadimitrious_Complexity_Book}. 
Without topological
constraints and increasing the degree of collaboration among
sensors, a diffusion RLS algorithm was proposed
in~\cite{Cattivelli_Lopes_Sayed_TSP_Diffusion_RLS}. In addition to
local estimates, sensors continuously diffuse raw sensor
observations and regression vectors per neighborhood. This
facilitates percolating new data across the WSN, but 
estimation performance is degraded in the presence of communication noise.
When both the sensor measurements and regression 
vectors are corrupted by additive (colored) noise, the
diffusion-based RLS algorithm of~\cite{Bertrand_Moonen_Sayed_TSP_Diffusion_BC_RLS} exploits sensor
cooperation to reduce bias in the EWLSE.
All~\cite{Cattivelli_Lopes_Sayed_TSP_Diffusion_RLS},~\cite{Bertrand_Moonen_Sayed_TSP_Diffusion_BC_RLS}
and~\cite{Lopes_Sayed_Incremental_RLS_2006} include steady-state
MSE performance analysis under the independence
setting assumptions~\cite[p. 448]{Sayed_Adaptive_Book}. Distributed least
mean-squares (LMS) counterparts are also available, trading off
computational complexity for estimation performance; for
noteworthy representatives
see~\cite{Lopes_Sayed_TSP_Diffusion_LMS,Yannis_Gonzalo_GG_DLMS,
Cattivelli_Sayed_TSP_Incremental_LMS}, and
references therein. Recent studies have also considered
more elaborate sensor processing strategies including 
projections~\cite{Li_Lopes_Chambers_Sayed_TSP_Incremental_LMS,
Chouvardas_Slavakis_Theoridis_TSP_Projections}, 
adaptive combination 
weights~\cite{Takahashi_Yamada_Sayed_TSP_LMS_Combiners}, 
or even sensor hierarchies~\cite{Yannis_Gonzalo_GG_DLMS,
Cattivelli_Sayed_SSP_2009}, and mobility~\cite{Tu_Sayed_JSTSP_Mobility}.

Several distributed (adaptive) estimation algorithms are rooted on iterative
optimization methods, which capitalize upon the separable structure of the
cost defining the desired estimator. The sample mean estimator was
formulated in~\cite{Rabbat_SPAWC_2005} as an optimization problem, and was
solved in a distributed fashion using a primal dual approach; see,
e.g.,~\cite{Bertsekas_Book_Distr}. Similarly, the incremental schemes
in e.g.,~\cite{Lopes_Sayed_Incremental_RLS_2006,Cattivelli_Sayed_TSP_Incremental_LMS,RNV'07,Rabbat_Incremental_Algorithms_2005} 
are all based on incremental
(sub)gradient methods~\cite{Nedic_Incremental_2001}. Even the
diffusion LMS algorithm of~\cite{Lopes_Sayed_TSP_Diffusion_LMS} has been
recently shown related to incremental strategies, when
these are adopted to optimize an approximate reformulation of the LMS
cost~\cite{Cattivelli_Sayed_TSP_Diffusion_LMS}. Building on the framework
introduced by~\cite{Yannis_Ale_GG_PartI},
the D-LMS and D-RLS algorithms in~\cite{Yannis_Gonzalo_GG_DLMS,
Gonzalo_Yannis_GG_DLMS_Perf,Gonzalo_Yannis_GG_DRLS} are obtained upon
recasting the respective decentralized estimation problems as multiple
equivalent constrained subproblems. The resulting minimization subtasks
are shown to be highly paralellizable across sensors, when carried out 
using the alternating-direction method of multipliers
(AD-MoM)~\cite{Bertsekas_Book_Distr}. Much related to the AD-MoM 
is the alternating minimization algorithm (AMA)~\cite{Tseng_AMA_1991},
used here to develop a novel D-RLS algorithm
offering reduced complexity when compared to
its counterpart of~\cite{Gonzalo_Yannis_GG_DRLS}.

% % % % % % % % % % % % % % % % % % % % % % % % % % % % % % % % % % % % % % % %
%                         Subsection I-B                                      %
% % % % % % % % % % % % % % % % % % % % % % % % % % % % % % % % % % % % % % % %

\subsection{Contributions and paper outline}
\label{ssec:contributions}
The present paper develops a fully distributed (D-) RLS type of
algorithm, which performs in-network, adaptive LS
estimation. D-RLS is applicable to general ad hoc WSNs that are
challenged by additive communication noise, and may lack a Hamiltonian
cycle altogether. Different
from the distributed Kalman trackers of e.g.,~\cite{Ale_Yannis_Stergios_GG_ControlsMag,Cattivelli_Sayed_TAC_Kalman}, 
the universality of the LS principle broadens the applicability of
D-RLS to a wide class of
distributed adaptive estimation tasks, since it requires no knowledge of the
underlying state space model. The algorithm is developed by reformulating the
EWLSE into an equivalent constrained form~\cite{Yannis_Ale_GG_PartI}, 
which can be minimized
in a distributed fashion by capitalizing on the separable structure
of the augmented Lagrangian using the 
AMA solver in~\cite{Tseng_AMA_1991} (Section \ref{sec:probstatement}). From 
an algorithmic standpoint, the novel distributed iterations here
offer two extra features relative to the AD-MoM-based D-RLS
variants in~\cite{DRLS_Asilomar_2007,Gonzalo_Yannis_GG_DRLS}. First,
as discussed in Section \ref{ssec:AD-MoM_DRLS} 
the per sensor computational complexity is markedly reduced, since there
is no need to explicitly carry out a matrix inversion per iteration
as in~\cite{Gonzalo_Yannis_GG_DRLS}. 
 Second, the approach here bypasses the need of the so-termed bridge 
 sensors~\cite{DRLS_Asilomar_2007}. As a result, a fully distributed 
 algorithm is obtained whereby all sensors perform
the same tasks in a more efficient manner, without introducing 
hierarchies that may require intricate recovery protocols 
to cope with sensor failures. 

Another contribution of the present paper pertains to a detailed stability
and MSE steady-state \textit{performance analysis} for D-RLS
(Section \ref{sec:Stability_Perf_Analysis}). These theoretical results 
were lacking in the algorithmic papers~\cite{DRLS_Asilomar_2007,Gonzalo_Yannis_GG_DRLS},
where claims were only supported via computer simulations. 
Evaluating the performance of (centralized) adaptive filters  
is a challenging problem in its own right; 
prior art is surveyed in
e.g.,~\cite{Solo_Stability_LMS},~\cite[pg. 120]{Solo_Adaptive_Book},
\cite[pg. 357]{Sayed_Adaptive_Book}, and the extensive list of references
therein. On top of that, a WSN setting introduces unique challenges in the
analysis such as space-time sensor data and multiple sources of additive
noise, a consequence of imperfect sensors and communication links. The
approach pursued here capitalizes on an `averaged' error-form representation of the
local recursions comprising D-RLS, as a global dynamical system described
by a stochastic difference-equation derived in Section \ref{ssec:analysis_prelim}.
The covariance matrix of the resulting state is then shown to encompass
all the information needed to evaluate the relevant global and
sensor-level performance metrics (Section \ref{ssec:fig_merit}).
For sensor observations that are linearly
related to the time-invariant parameter vector sought, 
the simplifying independence setting
assumptions~\cite[pg. 110]{Solo_Adaptive_Book},~\cite[pg.
448]{Sayed_Adaptive_Book} are key enablers towards deriving accurate
closed-form expressions for the mean-square deviation and
excess-MSE steady-state values (Section
\ref{ssec:SS_Perf}). Stability in the mean- and MSE-sense are also investigated,
 revealing easily-checkable sufficient conditions under 
which a steady-state is attained. 

Numerical tests corroborating the theoretical findings
are presented in Section
\ref{sec:sims_Perf}, while concluding remarks and possible directions
for future work are given in Section
\ref{sec:conc}.\\

% % % % % % % % % % % % % % % % % % % % % % % % % % % % % % % % % % % % % % % %
%                               Notation                                      %
% % % % % % % % % % % % % % % % % % % % % % % % % % % % % % % % % % % % % % % %

\noindent\textit{Notation:} Operators $\otimes$,
$(.)^{T}$, $(.)^{\dagger}$, $\lambda_{\max}(.)$,
$\mbox{tr}(.)$, $\mbox{diag}(.)$, $\mbox{bdiag}(.)$,
$E\left[.\right]$, $\textrm{vec}\left[.\right]$ will denote
Kronecker product, transposition, matrix
pseudo-inverse, spectral radius, matrix trace, diagonal matrix,
block diagonal matrix, expectation, and matrix vectorization,
respectively. For both vectors and matrices,
$\|.\|$ will stand for the $2-$norm. and $|.|$ for the
cardinality of a set or the magnitude
of a scalar. Positive definite matrix $\mathbf{M}$ will be
denoted by $\bbM\succ\mathbf{0}$. The $n\times n$ identity
matrix will be represented by $\mathbf{I}_{n}$, while
$\mathbf{1}_{n}$ will denote the $n\times 1$ vector of all ones
and $\mathbf{1}_{n\times m}:=\mathbf{1}_{n}\mathbf{1}_{m}^{T}$.
Similar notation will be adopted for vectors (matrices) of all
zeros. For matrix $\bbM\in\mathbb{R}^{m\times n}$,
$\textrm{nullspace}(\bbM):=\{\bbx\in\mathbb{R}^{n}:\bbM\bbx=\mathbf{0}_m\}$.
The $i$-th vector in the canonical basis for $\mathbb{R}^n$
will be denoted by $\bbb_{n,i}$, $i=1,\ldots,n$.

% % % % % % % % % % % % % % % % % % % % % % % % % % % % % % % % % % % % % % % %
%                         Section II                                          %
% % % % % % % % % % % % % % % % % % % % % % % % % % % % % % % % % % % % % % % %

\section{Problem Statement and Distributed RLS Algorithm}
\label{sec:probstatement} Consider a WSN with sensors $\{1,\ldots,
J\}:=\mathcal{J}$. Only single-hop communications are allowed,
i.e., sensor $j$ can communicate only with the sensors in its
neighborhood $\mathcal{N}_j\subseteq\mathcal{J}$, having
cardinality $|\calN_j|$. Assuming that inter-sensor links are
symmetric, the WSN is modeled as an undirected connected graph with
associated graph Laplacian matrix $\bbL$.
Different from
\cite{Cattivelli_Lopes_Sayed_TSP_Diffusion_RLS,Bertrand_Moonen_Sayed_TSP_Diffusion_BC_RLS} and
\cite{Lopes_Sayed_Incremental_RLS_2006}, the present network model
accounts explicitly for non-ideal sensor-to-sensor links.
Specifically, signals received at sensor $j$ from sensor $i$ at
discrete-time instant $t$ are corrupted by a zero-mean additive
noise vector $\bbeta_{j}^{i}(t)$, assumed temporally and spatially
uncorrelated. The communication noise covariance matrices are denoted by 
$\bbR_{\eta_j}:=E[\bbeta_{j}^{i}(t)(\bbeta_{j}^{i}(t))^T]$, $j\in\calJ$.

The WSN is deployed to estimate a real signal vector
$\bbs_0\in\mathbb{R}^{p\times 1}$ in a distributed fashion and
subject to the single-hop communication constraints, by resorting
to the LS criterion~\cite[p. 658]{Sayed_Adaptive_Book}. Per time
instant $t=0,1,\ldots,$ each sensor acquires a regression vector
$\bbh_j(t)\in\mathbb{R}^{p\times 1}$ and a scalar observation
$x_j(t)$, both assumed zero-mean without loss of generality. A
similar setting comprising complex-valued data was considered in
\cite{Cattivelli_Lopes_Sayed_TSP_Diffusion_RLS} and
\cite{Lopes_Sayed_Incremental_RLS_2006}. Here, the exposition focuses on
real-valued quantities for simplicity, but extensions to the complex case are 
straightforward. Given new data
sequentially acquired, a pertinent approach is to consider the
EWLSE
~\cite{Sayed_Adaptive_Book,Cattivelli_Lopes_Sayed_TSP_Diffusion_RLS,Lopes_Sayed_Incremental_RLS_2006}
\begin{equation}\label{estprblm}
\hhatbbs_{\textrm{ewls}}(t):=\mbox{arg}\:\min_{\bbs}\sum_{\tau=0}^{t}\sum_{j=1}^{J}
\lambda^{t-\tau}\left[x_{j}(\tau)-\bbh^{T}_{j}(\tau)\bbs\right]^{2}+
\lambda^{t}\bbs^{T}\bbPhi_0\bbs
\end{equation}
where $\lambda\in(0,1]$ is a forgetting factor, while 
$\bbPhi_0\succ \mathbf{0}_{p\times p}$ is included for regularization. Note
that in forming the EWLSE at time $t$, the entire history of data
$\{x_j(\tau),\bbh_j(\tau)\}_{\tau=0}^{t},$ $\forall\;j\in\calJ$ is
incorporated in the online estimation process. Whenever
$\lambda<1$, past data are exponentially discarded thus enabling
tracking of nonstationary processes. Regarding applications, a distributed power spectrum estimation
task matching the aforementioned problem statement, 
can be found in~\cite{Gonzalo_Yannis_GG_DRLS}.

To decompose the cost function in \eqref{estprblm}, in which
summands are coupled through the global variable $\bbs$, 
introduce auxiliary variables $\{\bbs_j\}_{j=1}^{J}$ 
representing local estimates of $\bbs_0$ per sensor $j$. These local
estimates are utilized to form the separable convex \emph{constrained}
minimization problem
\begin{align}
\{\hat{\mathbf{s}}_j(t)\}_{j=1}^{J}:=&\arg\min_{\{\mathbf{s}_j\}_{j=1}^{J}}
\sum_{\tau=0}^{t}\sum_{j=1}^{J}\lambda^{t-\tau}[
x_j(\tau)-\mathbf{h}_{j}^T(\tau)\mathbf{s}_j]^2+
J^{-1}\lambda^t\sum_{j=1}^{J}\mathbf{s}_j^T\mathbf{\Phi}_0\mathbf{s}_j,
\nonumber\\
&\textrm{s. t. }\; \mathbf{s}_j=
\mathbf{s}_{j'},\;\;j\in\calJ,\:\:j'\in\calN_j.\label{Decomp_RLS}
\end{align}
From the connectivity of the WSN, % it follows readily that
\eqref{estprblm} and \eqref{Decomp_RLS} are equivalent in the sense
that $\hat{\bbs}_j(t)=\hat{\bbs}_{\textrm{ewls}}(t)$, $\forall\: j\in\calJ$ and
 $t\geq 0$;
see also \cite{Yannis_Ale_GG_PartI}. To arrive at the D-RLS recursions, 
it is convenient to reparametrize the constraint set
 \eqref{Decomp_RLS} in the equivalent form
\begin{equation}\label{Constr_Equi}
\bbs_j=\bar{\bbz}_j^{j'},\:\:\bbs_{j'}=\tilde{\bbz}_j^{j'},\textrm{
  and  }\bar{\bbz}_j^{j'}=\tilde{\bbz}_j^{j'},\;\;
j\in\calJ,\;\;j'\in\calN_j,\;\;j\neq j'.
\end{equation}
where $\{\bar{\bbz}_j^{j'},\tilde{\bbz}_j^{j'}\}_{j'\in\calN_j}$, $j\in\calJ$, 
are auxiliary optimization variables that will be eventually eliminated.

% % % % % % % % % % % % % % % % % % % % % % % % % % % % % % % % % % % % % % % %
%                         Subsection II-A                                     %
% % % % % % % % % % % % % % % % % % % % % % % % % % % % % % % % % % % % % % % %

\subsection{The D-RLS algorithm}
\label{ssec:AMA_DRLS}

To tackle the constrained minimization problem 
\eqref{Decomp_RLS} at time instant $t$, associate
Lagrange multipliers $\bbv_{j}^{j'}$ and $\bbu_{j}^{j'}$ with the first
pair of
consensus constraints in \eqref{Constr_Equi}. Introduce the ordinary Lagrangian function
\begin{align}\label{Lagr}
\ccalL\left[\calbbs,\calbbz,\calbbv,\calbbu\right]={}&
\sum_{j=1}^{J}\sum_{\tau=0}^{t}\lambda^{t-\tau}[
x_j(\tau)-\mathbf{h}_{j}^T(\tau)\mathbf{s}_j]^2+J^{-1}\lambda^t\sum_{j=1}^{J}
\bbs_j^T\bbPhi_0\mathbf{s}_j\nonumber\\
%The constraints terms
&+\sum_{j=1}^{J}\sum_{j'\in\calN_{j}}\left[(\bbv_{j}^{j'})^{T}
(\bbs_{j}-\bar{\bbz}_j^{j'})+(\bbu_j^{j'})^T
(\bbs_{j'}-\tilde{\bbz}_j^{j'})\right]
\end{align}
as well as the quadratically \textit{augmented} Lagrangian 
\begin{equation}\label{augLagr}
\ccalL_c\left[\calbbs,\calbbz,\calbbv,\calbbu\right]=\ccalL\left[\calbbs,\calbbz,\bbv,\bbu\right]
+\frac{c}{2}\sum_{j=1}^{J}\sum_{j'\in\calN_{j}}\left[
\|\bbs_{j}-\bar{\bbz}_j^{j'}\|_2^2+\|\bbs_{j'}-\tilde{\bbz}_j^{j'}\|_2^2
\right]
\end{equation}
where $c$ is a positive penalty coefficient; and
$\calbbs:=\{\bbs_j\}_{j=1}^{J}$,
$\calbbz:=\{\bar{\bbz}_j^{j'},\tilde{\bbz}_j^{j'}\}_{j\in\calJ}^{j'\in\calN_j}$, and
$[\calbbv,\calbbu]:=\{\bbv_{j}^{j'},\bbu_{j}^{j'}\}^{j'\in\calN_j}_{j\in\calJ}$.
Observe that the remaining constraints in \eqref{Constr_Equi}, namely
$\calbbz\in C_z:=\{\calbbz\: :\:\bar{\bbz}_j^{j'}=\tilde{\bbz}_j^{j'},\;
j\in\calJ,\;j'\in\calN_j,\;j\neq j'\}$, have not been dualized.

Towards deriving the D-RLS recursions, the alternating minimization
algorithm (AMA) of~\cite{Tseng_AMA_1991} will be adopted here to tackle the 
separable EWLSE reformulation \eqref{Decomp_RLS} in a distributed fashion. Much
 related to AMA is 
the alternating-direction method of multipliers (AD-MoM), an iterative 
augmented Lagrangian method specially well-suited for parallel
processing 
~\cite{Yannis_Ale_GG_PartI,Gonzalo_Yannis_GG_DRLS,Bertsekas_Book_Distr}. While
the AD-MoM has been proven successful 
to tackle the optimization tasks stemming from general distributed estimators of deterministic and
(non-)stationary random signals, it is somehow curious that the AMA
has remained largely underutilized.

To minimize \eqref{Decomp_RLS} at time instant $t$, the AMA solver 
entails an iterative procedure comprising three steps per iteration
$k=0,1,2,\ldots$
\begin{description}
\item [{\bf [S1]}] \textbf{Multiplier updates:}\begin{align}
    \bbv_j^{j'}(t;k)&=\bbv_{j}^{j'}(t;k-1)+
c[\bbs_j(t;k)-\bar{\bbz}_j^{j'}(t;k)],
{\quad}j\in\calJ,{\:}j'\in\calN_{j}\nonumber\\
\bbu_j^{j'}(t;k)&=\bbu_{j}^{j'}(t;k-1)+
c[\bbs_{j'}(t;k)-\tilde{\bbz}_{j}^{j'}(t;k)],
{\quad}j\in\calJ,{\:}j'\in\calN_{j}.\nonumber
\end{align}

\item [{\bf [S2]}]  \textbf{Local estimate updates:}
    \begin{equation}\calbbs(t,k+1)=
    \mbox{arg}\:\min_{\calbbs}\ccalL\left[\calbbs,\calbbz(t,k),\calbbv(t,k),
    \calbbu(t,k)\right].
    \label{S2_AMA}\end{equation}

\item [{\bf [S3]}] \textbf{Auxiliary variable updates:}
    \begin{equation}\calbbz(t,k+1)=
    \mbox{arg}\:\min_{\calbbz\in 
C_z}\ccalL_{c}\left[\calbbs(t,k+1),\calbbz,\calbbv(t,k),\calbbu(t,k)\right
].
    \label{S3_AMA}\end{equation}
\end{description}
Steps [S1] and [S3] are identical to those in AD-MoM~\cite{Bertsekas_Book_Distr}. 
In words, these steps correspond to
dual ascent iterations to update the Lagrange multipliers, and
a block coordinate-descent minimization of the augmented Lagrangian with
respect to $\calbbz\in 
C_z$, respectively. The only
difference is with regards to the local estimate updates in [S2], where in
AMA the new iterates are obtained by minimizing the ordinary Lagrangian
with respect to $\calbbs$. For the sake of the aforementioned
minimization, all other variables are considered fixed taking their most
up to date values $\{\calbbz(t,k),\calbbv(t,k),\calbbu(t,k)\}$. For the
AD-MoM instead, the minimized quantity is the augmented Lagrangian both in [S2]
and in [S3]. 

The AMA was motivated in~\cite{Tseng_AMA_1991} for separable problems
that are strictly convex in $\calbbs$, but (possibly) only convex with respect 
to
$\calbbz$. Under this assumption, [S2] still yields a unique minimizer per iteration, 
and the AMA is useful for those cases in which the Lagrangian is much simpler
to optimize than the augmented Lagrangian. Because of the regularization matrix 
$\bbPhi_0\succ\mathbf{0}_{p\times p}$, the
EWLS cost in \eqref{Decomp_RLS} is indeed strictly convex for all
$t>0$, and the AMA is applicable. Section \ref{ssec:AD-MoM_DRLS} dwells into 
the benefits of minimizing the 
ordinary Lagrangian instead of its augmented counterpart \eqref{augLagr}, in 
the context
of distributed RLS estimation.

Carrying out the minimization in [S3] first, one finds
\begin{equation*}
\bar{\bbz}_{j}^{j'}(t,k+1)=\tilde{\bbz}_{j}^{j'}(t,k+1)=\frac{1}{2}\left[
\bbs_j(t,k+1)+\bbs_{j'}(t,k+1)\right]
, {\quad}j\in\calJ,{\:}j'\in\calN_{j}
\end{equation*}
so that $\bbv_j^{j'}(t;k)=-\bbu_j^{j'}(t;k)$ for all 
$k>-1$~\cite{Gonzalo_Yannis_GG_DRLS}. As a result
$\bbv_j^{j'}(t;k)$ is given by 
\begin{equation}\label{Vj_Update}
\bbv_{j}^{j'}(t;k)=\bbv_{j}^{j'}(t;k-1)
+\frac{c}{2}\left[\bbs_j(t;k)-\bbs_{j'}(t;k)\right],\;\;
j\in\calJ,\;\;j'\in{\cal N}_j.
\end{equation}
Moving on to [S2],
from the separable structure of \eqref{Lagr} the minimization
\eqref{S2_AMA} can be split into $J$ subproblems
\begin{equation*}
\bbs_j(t,k+1)=\mbox{arg}\:\min_{\bbs_j}\left[\sum_{\tau=0}^{t}\lambda^{t
-\tau}[
x_j(\tau)-\mathbf{h}_{j}^T(\tau)\mathbf{s}_j]^2+J^{-1}\lambda^{t}\bbs_j^T
\bbPhi_0\mathbf{s}_j
+\sum_{j'\in\calN_{j}}\left[\bbv_{j}^{j'}(t,k)
-\bbv^{j}_{j'}(t,k)\right]^{T}\bbs_{j}\right].
\end{equation*}
Since each of the local subproblems corresponds to an unconstrained quadratic
minimization, they all admit closed-form solutions
\begin{equation}
\bbs_j(t,k+1)=\bbPhi_j^{-1}(t)\bbpsi_j(t)-\frac{1}{2}\bbPhi_j^{-1}(t)
\sum_{j'\in\calN_j}\left[\bbv_{j}^{j'}(t,k)
-\bbv^{j}_{j'}(t,k)\right] \label{sjupdate_AMA}
\end{equation}
where
\begin{align}
\mathbf{\Phi}_j(t)&:=\sum_{\tau=0}^{t}\lambda^{t-\tau}\mathbf{h}_j(\tau)\mathbf{h}_j^T(\tau)+J^{-1}\lambda^{t}
\mathbf{\Phi}_0=\lambda \mathbf{\Phi}_j(t-1)+\mathbf{h}_j(t)\mathbf{h}_j^T(t)\label{Phi_update_AMA}\\
\bbpsi_j(t)&:=\sum_{\tau=0}^{t}\lambda^{t-\tau}\mathbf{h}_j(\tau)x_j(\tau)=
\lambda\bbpsi_j(t-1)+\bbh_j(t)x_j(t)\label{psi_update_AMA}.
\end{align}
Recursions \eqref{Vj_Update} and
\eqref{sjupdate_AMA} constitute the AMA-based D-RLS algorithm,
whereby all sensors  $j\in\calJ$ keep track of their local estimate
$\bbs_j(t;k+1)$ and their multipliers
$\{\bbv_j^{j'}(t;k)\}_{j'\in\calN_j}$, which can be arbitrarily
initialized. From the rank-one update in \eqref{Phi_update_AMA} 
and capitalizing on the matrix inversion lemma, matrix
$\mathbf{\Phi}_j^{-1}(t)$ can be efficiently updated according to
\begin{equation}\label{Updating_Phij}
\bbPhi_j^{-1}(t)= \lambda^{-1}\bbPhi_j^{-1}(t-1)-
\frac{\lambda^{-1}\bbPhi_j^{-1}(t)\bbh_j(t)\bbh_j^T(t)\bbPhi_j^{-1}(t-1)}
{\lambda+\bbh_j^T(t)\bbPhi_j^{-1}(t-1)\bbh_j(t)}.
\end{equation}
with complexity ${\cal O}(p^2)$. It is recommended to initialize the
matrix recursion with
$\mathbf{\Phi}_j^{-1}(0)=J\bbPhi_0^{-1}:=\delta\bbI_p$, where $\delta>0$
is chosen sufficiently large~\cite{Sayed_Adaptive_Book}. Not surprisingly,
by direct application of the convergence results in~\cite[Proposition
3]{Tseng_AMA_1991}, it follows that:

\begin{proposition}\label{prop_2}
For arbitrarily initialized
$\{\bbv_j^{j'}(t;-1)\}_{j\in\calJ}^{j'\in\calN_j}$, $\bbs_j(t;0)$ and
$c\in(0,c_{u})$; the local estimates $\bbs_j(t;k)$ generated by
\eqref{sjupdate_AMA} reach consensus as $k\rightarrow\infty$; i.e.,
$$\lim_{k\rightarrow\infty}\bbs_j(t;k)=\hat{\bbs}_{\textrm{ewls}}(t),\;\textrm{
for all }j\in\calJ.$$
\end{proposition}
The upper bound $c_u$ is proportional to the
modulus of the strictly convex cost function in \eqref{Decomp_RLS}, and
inversely proportional to the norm of a matrix suitably chosen to express
the linear constraints in \eqref{Constr_Equi}; further details are 
in~\cite[Section 4]{Tseng_AMA_1991}. Proposition \ref{prop_2} asserts that per 
time instant $t$, the AMA-based D-RLS algorithm yields a sequence of local 
estimates that converge to the global EWLSE sought, as $k\to\infty$, 
or, pragmatically for large enough $k$. In principle, one could argue that running many 
consensus iterations may not be a problem in a stationary environment. However, 
when the WSN is deployed to track a time-varying parameter 
vector $\bbs_0(t)$, one cannot afford significant delays in-between 
consecutive sensing instants. 

One possible way to overcome this hurdle is to run a single consensus iteration
 per acquired observation
$x_{j}(t)$. Specifically, letting $k=t$ in recursions
\eqref{Vj_Update}-\eqref{sjupdate_AMA}, one arrives at a single time scale
D-RLS algorithm which is suitable for operation in nonstationary WSN
environments. Accounting also for additive communication noise that
corrupts the exchanges of multipliers and local estimates, the per sensor
tasks comprising the novel AMA-based \textit{single time scale} D-RLS algorithm
 are given by
\boxedeqnmulti{\mathbf{v}_j^{j'}(t)={}&\mathbf{v}_{j}^{j'}(t-1)
+\frac{c}{2}\left[\mathbf{s}_{j}(t)-({\mathbf{s}}_{j'}(t)+\bbeta_j^{j'}(t))
\right],
{\quad}j'\in\calN_j\label{ST_Vj_Update_AMA}\\
\bbPhi_j^{-1}(t+1)={}& \lambda^{-1}\bbPhi_j^{-1}(t)-
\frac{\lambda^{-1}\bbPhi_j^{-1}(t)\bbh_j(t+1)\bbh_j^T(t+1)\bbPhi_j^{-1}(t)}
{\lambda+\bbh_j^T(t+1)\bbPhi_j^{-1}(t)\bbh_j(t+1)}\label{ST_Phij_Update_AMA}\\
\bbpsi_j(t+1)={}&\lambda\bbpsi_j(t)+\bbh_j(t+1)x_j(t+1)
\label{ST_psij_Update_AMA}\\
\bbs_j(t+1)={}&\bbPhi_j^{-1}(t+1)\bbpsi_j(t+1)
-\frac{1}{2}\bbPhi_j^{-1}(t+1)\sum_{j'\in\calN_j}\left[\bbv_{j}^{j'}(t)-
(\bbv_{j'}^j(t)+\bar{\bbeta}_{j}^{j'}(t))\right].\label{ST_Sj_Update_AMA}}
Recursions \eqref{ST_Vj_Update_AMA}-\eqref{ST_psij_Update_AMA} are tabulated as 
Algorithm \ref{AMA_STD-RLS_algorithm_table}, which also details the inter-sensor
communications of multipliers and local estimates taking place within neighborhoods. When
powerful error control codes render inter-sensor links virtually ideal,
direct application of the results in~\cite{Gonzalo_Yannis_GG_DRLS,Gonzalo_Yannis_GG_DLMS_Perf} 
show that D-RLS can be further simplified to reduce the communication overhead and memory
storage requirements.
\setlength{\arraycolsep}{5pt}
%%%%%%%%%%%
%%%%%%%%%%%
\begin{algorithm}[t]
\caption{: AMA-based D-RLS} \small{
\begin{algorithmic}
    \STATE Arbitrarily initialize $\{\bbs_j(0)\}_{j=1}^{J}$
    and $\{\bbv_{j}^{j'}(-1)\}_{j\in\calJ}^{j'\in{\calN}_{j}}$.
    \FOR {$t=0,1$,$\ldots$}
        \STATE All $j\in\calJ$: transmit $\bbs_{j}(t)$ to neighbors in 
$\calN_j$.
        \STATE All $j\in\calJ$: update $\{\bbv^{j'}_j(t)\}_{j'\in\calN_j}$ 
using \eqref{ST_Vj_Update_AMA}.
        \STATE All $j\in\calJ$: transmit $\bbv^{j'}_j(t)$ to each 
$j'\in\calN_{j}$.
		\STATE All $j\in\calJ$: update $\bbPhi_{j}(t+1)$ and $\bbpsi_j(t+1)$ 
using 
		\eqref{ST_Phij_Update_AMA} and \eqref{ST_psij_Update_AMA}, 
respectively.
        \STATE All $j\in\calJ$: update $\bbs_{j}(t+1)$ using 
\eqref{ST_Sj_Update_AMA}.
    \ENDFOR
\end{algorithmic}}
\label{AMA_STD-RLS_algorithm_table}
\end{algorithm}
%

% % % % % % % % % % % % % % % % % % % % % % % % % % % % % % % % % % % % % % % %
%                         Subsection II-B                                     %
% % % % % % % % % % % % % % % % % % % % % % % % % % % % % % % % % % % % % % % %

\subsection{Comparison with the AD-MoM-based D-RLS algorithm}
\label{ssec:AD-MoM_DRLS}

A related D-RLS algorithm was put forth in~\cite{Gonzalo_Yannis_GG_DRLS}, 
whereby the decomposable exponentially-weighted LS cost \eqref{Decomp_RLS} is 
minimized using the AD-MoM, rather than the AMA as in Section 
\ref{ssec:AMA_DRLS}. Recall that the AD-MoM solver yields $\bbs_j(t+1)$ as the 
optimizer of the augmented 
Lagragian, while its AMA counterpart minimizes 
the ordinary Lagrangian instead. Consequently, different from 
\eqref{ST_Sj_Update_AMA} local estimates 
in the AD-MoM-based D-RLS algorithm of~\cite{Gonzalo_Yannis_GG_DRLS} are 
updated via
\begin{align}
\bbs_j(t+1)&=\bar\bbPhi_j^{-1}(t+1)\bbpsi_j(t+1)
+\frac{c}{2}\bar\bbPhi_j^{-1}(t+1)\sum_{j'\in\calN_j}\left[
\bbs_j(t)+(\bbs_{j'}(t)+\bbeta_{j}^{j'}(t))\right]
\nonumber\\
&\hspace{2.8cm}-\frac{1}{2}\bar\bbPhi_j^{-1}(t+1)\sum_{j'\in\calN_j}\left[\bbv_
{j}^{
j'}(t)-
(\bbv_{j'}^j(t)+\bar{\bbeta}_{j}^{j'}(t))\right]\label{ST_Sj_Update_ADMoM}
\end{align}
where [cf. \eqref{Phi_update_AMA}]
\begin{equation}\label{Phij_def_ADMoM}
\bar\bbPhi_j(t):=\sum_{\tau=0}^{t}\lambda^{t-\tau}\mathbf{h}_j(\tau)
\mathbf{h}_j^T(\tau)+J^{-1}\lambda^{t}\mathbf{\Phi}_0+c|\calN_j|\bbI_p.
\end{equation}
Unless $\lambda=1$, it is impossible to derive a rank-one update for 
$\bar{\bbPhi}_j(t)$ as in \eqref{Phi_update_AMA}. The reason is  the 
regularization term $c|\calN_j|\bbI_p$ in
\eqref{Phij_def_ADMoM}, a direct consequence of the quadratic penalty in the
augmented Lagrangian \eqref{augLagr}. This prevents one from efficiently 
updating 
$\bar\bbPhi_j^{-1}(t+1)$ in \eqref{ST_Sj_Update_ADMoM}  using the 
matrix inversion lemma [cf. \eqref{ST_Phij_Update_AMA}]. Direct inversion of
$\bar\bbPhi_j(t+1)$ per iteration dominates the computational
complexity of the AD-MoM-based D-RLS algorithm, which is roughly $\mathcal{O}(p^3)$~\cite{Gonzalo_Yannis_GG_DRLS}. 

Unfortunately, the 
penalty coefficient cannot be set to zero because the D-RLS algorithm breaks 
down. For instance, when the
initial Lagrange multipliers are null and $c=0$, D-RLS boils down to
a purely local (L-) RLS algorithm where sensors do not cooperate, hence consensus cannot 
be attained. 
All in all, the novel AMA-based D-RLS algorithm of this paper
offers an improved alternative with an order of magnitude reduction in terms of computational
complexity per sensor. With
regards to communication cost, the AD-MoM-based D-RLS and Algorithm 1 here incur
identical overheads; see~\cite[Sec. III-B]{Gonzalo_Yannis_GG_DRLS} for a detailed
analysis of the associated cost, as well as comparisons with the I-RLS~\cite{Lopes_Sayed_Incremental_RLS_2006} 
and diffusion RLS algorithms~\cite{Cattivelli_Lopes_Sayed_TSP_Diffusion_RLS}.

While the AMA-based D-RLS algorithm
is less complex computationally than its AD-MoM counterpart 
in~\cite{Gonzalo_Yannis_GG_DRLS}, 
Proposition \ref{prop_2} asserts that when many consensus iterations can be 
afforded, convergence to the centralized EWLSE is guaranteed provided 
$c\in(0,c_u)$.  On
the other hand, the AD-MoM-based D-RLS algorithm will attain the EWLSE for
any $c>0$ (cf.~\cite[Prop. 1]{Gonzalo_Yannis_GG_DRLS}). In addition, it does not require tuning
the extra parameter $\delta$, since it is applicable when
$\bbPhi_0=\mathbf{0}_{p\times p}$ because the augmented Lagrangian provides the
 needed
regularization. 

% % % % % % % % % % % % % % % % % % % % % % % % % % % % % % % % % % % % % % % %
%                         Section III                                         %
% % % % % % % % % % % % % % % % % % % % % % % % % % % % % % % % % % % % % % % %

\section{Analysis Preliminaries}
\label{sec:analysis_prelim}

% % % % % % % % % % % % % % % % % % % % % % % % % % % % % % % % % % % % % % % %
%                         Subsection III-A                                    %
% % % % % % % % % % % % % % % % % % % % % % % % % % % % % % % % % % % % % % % %

\subsection{Scope of the analysis: assumptions and approximations}
\label{ssec:scope}
Performance evaluation of the D-RLS algorithm is much more involved than
that of e.g., D-LMS~\cite{Yannis_Gonzalo_GG_DLMS,Gonzalo_Yannis_GG_DLMS_Perf}. 
The challenges are well documented for the classical
(centralized) LMS and RLS
filters~\cite{Sayed_Adaptive_Book,Solo_Adaptive_Book}, and results for the
latter are less common and typically involve simplifying approximations.
What is more, the distributed setting  introduces
unique challenges in the analysis. These include space-time sensor data
and multiple sources of additive noise, a consequence of imperfect sensors
and communication links. 

In order to proceed, a few typical modeling assumptions are introduced to
delineate the scope of the ensuing stability and performance results. For
all $j\in\calJ$, it is assumed that:
\begin{description}
\item [{\bf (a1)}] \emph{Sensor observations adhere to the linear model
    $x_j(t)=\mathbf{h}_j^T(t)\mathbf{s}_0+\epsilon_j(t)$, where the
    zero-mean white noise $\{\epsilon_j(t)\}$ has variance
    $\sigma_{\epsilon_j}^2$;}
\item [{\bf (a2)}] Vectors \emph{$\{\bbh_{j}(t)\}$ are
    spatio-temporally white with covariance matrix
    $\bbR_{h_j}\succ\mathbf{0}_{p\times p}$; and}
\item [{\bf (a3)}] Vectors \emph{$\{\bbh_{j}(t)\}$,
    $\{\epsilon_j(t)\}$, $\{\bbeta_j^{j'}(t)\}_{j'\in\calN_j}$ and
    $\{\bar{\bbeta}_j^{j'}(t)\}_{j'\in\calN_j}$ are independent.}
\end{description}
Assumptions (a1)-(a3) comprise the widely adopted \textit{independence
setting}, for sensor observations that are linearly related to the
time-invariant parameter of interest; see e.g.,~\cite[pg.
110]{Solo_Adaptive_Book},~\cite[pg. 448]{Sayed_Adaptive_Book}. Clearly,
(a2) can be violated in, e.g., FIR filtering of signals (regressors) with
a shift structure as in the distributed power spectrum estimation problem
described in~\cite{Yannis_Gonzalo_GG_DLMS} and~\cite{Gonzalo_Yannis_GG_DRLS}. 
Nevertheless, the
steady-state performance results extend accurately to the pragmatic setup
that involves time-correlated sensor data; see also the numerical tests in
Section \ref{sec:sims_Perf}. In line with a distributed setting such
 as a WSN, the statistical profiles of both  
regressors and the noise quantities vary across sensors (space), yet they are 
assumed to remain time invariant. For a related analysis of a distributed LMS 
algorithm operating in a 
nonstationary environment, the reader is referred 
to~\cite{Gonzalo_Yannis_GG_DLMS_Perf}.

In the particular case of the D-RLS algorithm, a unique challenge stems from
 the stochastic matrices $\bbPhi_j^{-1}(t)$ present in the
local estimate updates \eqref{ST_Sj_Update_AMA}. Recalling 
\eqref{Phi_update_AMA},
it is apparent that $\bbPhi_j^{-1}(t)$ depends upon the \textit{whole
history} of local regression vectors $\{\bbh_j(\tau)\}_{\tau=0}^t$. Even
obtaining $\bbPhi_j^{-1}(t)$'s distribution or computing its expected
value is a formidable task in general, due to the matrix inversion
operation. It is for these reasons that some simplifying approximations
will be adopted in the sequel, to carry out the analysis that otherwise
becomes intractable.

Neglecting the regularization term in \eqref{Phi_update_AMA} that vanishes
exponentially as $t\to\infty$, the matrix $\mathbf{\Phi}_j(t)$ is obtained
as an exponentially weighted moving average (EWMA). The EWMA can be seen
as an average modulated by a sliding window of equivalent length
$1/(1-\lambda)$, which clearly grows as $\lambda\to 1$. This observation
in conjunction with (a2) and the strong law of large numbers, justifies
the approximation
\begin{equation}\label{Phij_Averaged}
\mathbf{\Phi}_j(t)\approx
E[\mathbf{\Phi}_j(t)]=\frac{\bbR_{h_j}}{1-\lambda}, {\quad}0\ll\lambda<1\textrm{  and  }t\to\infty.
\end{equation}
The expectation of $\mathbf{\Phi}_j^{-1}(t)$, on the other hand, is 
considerably
harder to evaluate. To overcome this challenge, the
following approximation will be invoked~\cite{Sayed_Adaptive_Book,
Cattivelli_Lopes_Sayed_TSP_Diffusion_RLS}
\begin{equation}\label{invPhij_Averaged}
E[\mathbf{\Phi}_j^{-1}(t)]\approx E^{-1}[\mathbf{\Phi}_j(t)]
\approx(1-\lambda)\bbR_{h_j}^{-1}, {\quad}0\ll\lambda<1\textrm{  and  }t\to\infty.
\end{equation}
It is admittedly a crude approximation at first sight, because
$E\left[X^{-1}\right]\neq E[X]^{-1}$ in general, for any random variable $X$. 
However,
experimental evidence suggests that the approximation is sufficiently
accurate for all practical purposes, when the forgetting factor approaches
unity~\cite[p. 319]{Sayed_Adaptive_Book}.

%As discussed in this section, the
%resulting estimation error covariance matrix encompasses all the
%information needed to evaluate the relevant performance metrics; namely
%MSE, excess mean-square error (EMSE) and mean-square deviation (MSD). The
%aforementioned figures of merit ultimately assess the performance of
%D-LMS, both on a per-sensor basis and collectively by considering the WSN
%as a whole.

% % % % % % % % % % % % % % % % % % % % % % % % % % % % % % % % % % % % % % % %
%                         Subsection III-B                                    %
% % % % % % % % % % % % % % % % % % % % % % % % % % % % % % % % % % % % % % % %

\subsection{Error-form D-RLS}
\label{ssec:analysis_prelim} 
The approach here to steady-state performance analysis relies on an `averaged'
error-form system representation of D-RLS in 
\eqref{ST_Vj_Update_AMA}-\eqref{ST_Sj_Update_AMA}, where 
$\mathbf{\Phi}_j^{-1}(t)$
in \eqref{ST_Sj_Update_AMA} is replaced by the approximation
$(1-\lambda)\bbR_{h_j}^{-1}$, for sufficiently large $t$. Somehow related
approaches were adopted in~\cite{Cattivelli_Lopes_Sayed_TSP_Diffusion_RLS}
and~\cite{Bertrand_Moonen_Sayed_TSP_Diffusion_BC_RLS}. Other noteworthy analysis 
techniques include the energy-conservation methodology
in~\cite{Sayed_Energy_Conservation},~\cite[p. 287]{Sayed_Adaptive_Book},
and stochastic averaging~\cite[p. 229]{Solo_Adaptive_Book}. For
performance analysis of distributed adaptive algorithms seeking
time-invariant parameters, the former has been applied
in e.g.,~\cite{Lopes_Sayed_TSP_Incremental},
\cite{Lopes_Sayed_TSP_Diffusion_LMS},
while the latter can be found in~\cite{Yannis_Gonzalo_GG_DLMS}.

Towards obtaining such error-form representation, introduce the 
local estimation errors
$\{\bby_{1,j}(t):=\bbs_j(t)-\bbs_0\}_{j=1}^J$ and multiplier-based
quantities
$\{\bby_{2,j}(t):=\frac{1}{2}\sum_{j'\in\calN_j}(\bbv_j^{j'}(t-1)-\bbv^j_{j'}(t-1))\}_{j=1}^J$.
It turns out that a convenient global state to describe the spatio-temporal
dynamics of D-RLS in \eqref{ST_Vj_Update_AMA}-\eqref{ST_Sj_Update_AMA} is
$\bby(t):=[\bby_{1}^T(t)\;\bby_{2}^T(t)]^T=[\bby_{1,1}^T(t)\ldots\bby_{1,J}^T
(t)\:
\bby_{2,1}^T(t)\ldots\bby_{2,J}^T(t)]^T\in\mathbb{R}^{2Jp}$. In addition, to concisely capture the
 effects
of both observation and communication noise on the estimation errors
across the WSN, define the $Jp\times 1$ noise supervectors
$\bbepsilon(t):=\sum_{\tau=0}^{t}\lambda^{t-\tau}[\bbh_1^T(\tau)\epsilon_1(\tau)\ldots\bbh_J^T(\tau)\epsilon_J(\tau)]^T$
and $\bar{\bbeta}(t):=[\bar{\bbeta}_{1}^{T}(t)\ldots
\bar{\bbeta}_{J}^{T}(t)]^{T}$. Vectors
$\{\bar{\bbeta}_{j}(t)\}_{j=1}^{J}$ represent the aggregate noise
corrupting the multipliers received by sensor $j$ at time instant $t$, and are given by
\begin{equation}\label{baretajt_Ch4}
\bar{\bbeta}_{j}(t):=\frac{1}{2}
\sum_{j'\in\calN_j}\bar{\bbeta}_j^{j'}(t).
\end{equation}
Their respective covariance matrices are easily computable under
(a2)-(a3). For instance,
\begin{equation}\label{epsilon_cov}
\bbR_{\bbepsilon}(t):=E[\bbepsilon(t)\bbepsilon^{T}(t)]=\frac{1-\lambda^{2(t+1)}}{1-\lambda^2}
\textrm{bdiag}(\bbR_{h_1}\sigma_{\epsilon_1}^2,\ldots,\bbR_{h_J}\sigma_{\epsilon_J}^2)
\end{equation}
while the structure of
$\bbR_{\bar{\bbeta}}:=E[\bar{\bbeta}(t)\bar{\bbeta}^{T}(t)]$ is given in
Appendix E. Two additional $Jp\times 1$ communication noise
supervectors are needed, namely
$\bbeta_\alpha(t):=\left[(\bbeta_1^\alpha(t))^{T}\ldots
(\bbeta_J^\alpha(t))^{T}\right]^{T}$ and
$\bbeta_\beta(t):=\left[(\bbeta_1^\beta(t))^{T}\ldots(\bbeta_J^\beta(t))^{T}\right]^{T}$,
where for $j\in\calJ$
\begin{equation}\label{noises_alpha_beta}
\bbeta_j^\alpha(t):=\frac{c}{4}\sum_{j'\in\calN_j}\bbeta_{j}^{j'}(t),{\quad}
\bbeta_j^\beta(t):=\frac{c}{4}\sum_{j'\in\calN_j}\bbeta^{j}_{j'}(t).
\end{equation}
Finally, let $(c/2)\bbL\otimes\bbI_p\in\mathbb{R}^{Jp\times Jp}$ be a
matrix capturing the WSN connectivity pattern through the (scaled) graph
Laplacian matrix $\bbL$, and define
$\bbR_{h}^{-1}:=\textrm{bdiag}(\bbR_{h_1}^{-1},\ldots,\bbR_{h_J}^{-1})$.
Based on these definitions, it is possible to state the following
important lemma established in Appendix A.

\begin{lemma}\label{Lemma_1_Ch4}
Let (a1) and (a2) hold. Then for $t\geq t_0$ with $t_0$ sufficiently large
while $0\ll\lambda<1$, the global state $\bby(t)$ approximately evolves
according to
\begin{align}\label{ystate_DRLS}
\bby(t+1)=\mbox{bdiag}((1-\lambda)\bbR_{h}^{-1},\bbI_{Jp})&\left\{\bbUpsilon\bby(t)+\left[\begin{array}{c} 
\bbI_{Jp}\\
\mathbf{0}_{Jp\times Jp}
\end{array}\right]\bbepsilon(t+1)+\left[\begin{array}{c} \bbI_{Jp}\\
\mathbf{0}_{Jp\times Jp}
\end{array}\right]\bar{\bbeta}(t)\right.\nonumber\\
&\hspace{0.4cm}\left.
+\left[\begin{array}{c} \bbI_{Jp}\\
-\bbI_{Jp}
\end{array}\right]\bbeta_\alpha(t)-
\left[\begin{array}{c} \bbI_{Jp}\\
-\bbI_{Jp}
\end{array}\right]\bbeta_\beta(t)\right\}
\end{align}
where the $2Jp\times 2Jp$ matrix $\bbUpsilon$ consists of the $Jp\times
Jp$ blocks $[\bbUpsilon]_{11}=-[\bbUpsilon]_{21}=-\bbL_{c}$ and
$[\bbUpsilon]_{12}=-[\bbUpsilon]_{22}=-\bbI_{Jp}$. The initial condition
$\bby(t_0)$ should be selected as
$\bby(t_0)=\textrm{bdiag}(\bbI_{Jp},\bbL_c)\bby'(t_0)$, where $\bby'(t_0)$
is any vector in $\mathbb{R}^{2Jp}$.
\end{lemma}

The convenience of representing $\bby(t)$ as in Lemma \ref{Lemma_1_Ch4}
will become apparent in the sequel, especially when investigating sufficient conditions
under which the D-RLS algorithm is stable in the mean sense (Section \ref{ssec:DRLS_Mean_Stability}).
In addition, the covariance matrix of the state vector $\bby(t)$ can be shown
to encompass all the information needed to evaluate the relevant
per sensor and networkwide performance figures of merit, the subject dealt with
next.

% % % % % % % % % % % % % % % % % % % % % % % % % % % % % % % % % % % % % % % %
%                         Subsection III-C                                    %
% % % % % % % % % % % % % % % % % % % % % % % % % % % % % % % % % % % % % % % %

\subsection{Performance Metrics}
\label{ssec:fig_merit} When it comes to performance evaluation of adaptive
algorithms, it is customary to consider as figures of merit the so-called
MSE, excess mean-square error (EMSE), and mean-square deviation (MSD)
~\cite{Sayed_Adaptive_Book},~\cite{Solo_Adaptive_Book}. In the present
setup for distributed adaptive estimation, it is pertinent to address both
global (network-wide) and local (per-sensor)
performance~\cite{Lopes_Sayed_TSP_Diffusion_LMS}. After recalling the
definitions of the local a priori error $e_{j}(t):=
x_{j}(t)-\bbh_{j}^{T}(t)\bbs_{j}(t-1)$ and local estimation error
$\bby_{1,j}(t):=\bbs_j(t)-\bbs_0$, the per-sensor performance metrics
are defined as
\begin{align}
\textrm{MSE}_j(t)&:=E[e_j^2(t)]\nonumber\\
\textrm{EMSE}_j(t)&:=E[(\bbh_j^{T}(t)\bby_{1,j}(t-1))^2]\nonumber\\
\textrm{MSD}_j(t)&:=E[\|\bby_{1,j}(t)\|^2]\nonumber
\end{align}
whereas their global counterparts are defined as the respective averages
across sensors, e.g., $\textrm{MSE}(t):=J^{-1}\sum_{j=1}^{J}E[e_j(t)^2]$,
and so on. 

Next, it is shown that it suffices to evaluate the state covariance matrix
$\bbR_{y}(t):=E[\bby(t)\bby^{T}(t)]$ in order to assess the aforementioned
performance metrics. To this end, note that by virtue of (a1) it is
possible to write $e_j(t)=-\bbh_j^{T}(t)\bby_{1,j}(t-1)+\epsilon_j(t)$.
Because $\bby_{1,j}(t-1)$ is independent of the zero-mean
$\{\bbh_j(t),\epsilon_j(t)\}$ under (a1)-(a3), from the previous
relationship between the a priori and estimation errors one finds that
$\textrm{MSE}_j(t)=\textrm{EMSE}_j(t)+\sigma_{\epsilon_j}^2$. Hence, it
suffices to focus on the evaluation of $\textrm{EMSE}_j(t)$, through which
$\textrm{MSE}_j(t)$ can also be determined under the assumption that the 
observation noise variances are known, or can be estimated for that matter. If 
$\bbR_{y_{1,j}}(t):=
E[\bby_{1,j}(t)\bby_{1,j}^{T}(t)]$ denotes the $j$-th local error
covariance matrix, then
$\textrm{MSD}_j(t)=\textrm{tr}(\bbR_{y_{1,j}}(t))$; and under (a1)-(a3), a
simple manipulation yields
%$\textrm{EMSE}_j(t)=\textrm{tr}(\bbR_{h_j}\bbR_{y_{1,j}}(t-1)).$
%
\begin{align}
\textrm{EMSE}_j(t)&=E[\textrm{tr}((\bbh_j^{T}(t)\bby_{1,j}(t-1))^2)]
=\textrm{tr}(E[\bbh_j(t)\bbh_j^{T}(t)\bby_{1,j}(t-1)\bby_{1,j}^{T}(t-1)])\nonumber\\
&=\textrm{tr}(E[\bbh_j(t)\bbh_j^{T}(t)]E[\bby_{1,j}(t-1)\bby_{1,j}^{T}(t-1)])=
\textrm{tr}(\bbR_{h_j}\bbR_{y_{1,j}}(t-1)).\nonumber
\end{align}
To derive corresponding formulas for the global performance figures of merit, 
let $\bbR_{y_1}(t):=E[\bby_1(t)\bby_1^{T}(t)]$ denote the global error 
covariance matrix, and define
$\bbR_h:=E[\bbR_h(t)]=\textrm{bdiag}(\bbR_{h_1},\ldots,\bbR_{h_J})$. It follows
 that $\textrm{MSD}(t)=J^{-1}
\textrm{tr}(\bbR_{y_1}(t))$, and $\textrm{EMSE}(t)=
J^{-1}\textrm{tr}(\bbR_h\bbR_{y_1}(t-1))$.

It is now straightforward to recognize that $\bbR_y(t)$ indeed provides all 
the information needed to evaluate the performance of the D-RLS algorithm. For 
instance, observe that the global error
 covariance matrix 
$\bbR_{y_1}(t)$ corresponds to the
$Jp\times Jp$ upper left submatrix of $\bbR_{y}(t)$, which is denoted by 
$[\bbR_{y}(t)]_{11}$.
Further, the $j$-th $p\times p$ diagonal submatrix ($j=1,\ldots,J$)
of $[\bbR_{y}(t)]_{11}$ is exactly $\bbR_{y_{1,j}}(t)$, and is likewise denoted
 by $[\bbR_{y}(t)]_{11,j}$. For clarity, the aforementioned notational 
conventions regarding submatrices within $\bbR_y(t)$ are 
illustrated in Fig. \ref{blocks}.  In a nutshell, deriving a
closed-form expression for $\bbR_{y}(t)$ enables the evaluation of all
performance metrics of interest, as summarized in Table 
\ref{table:figures_of_merit}. This task will be
considered in 
Section \ref{ssec:SS_Perf}.

\begin{table}[t]
\renewcommand{\arraystretch}{1.3}
\caption{Evaluation of local and global figures of merit from $\bbR_y(t)$}
\label{table:figures_of_merit} \centering
\begin{tabular}{|c||c|c|c|}
\hline
 & \bfseries MSD &  \bfseries EMSE & \bfseries MSE \\
\hline\hline
\bfseries Local & $\textrm{tr}([\bbR_{y}(t)]_{11,j})$ & 
$\textrm{tr}(\bbR_{h_j}[\bbR_{y}(t-1)]_{11,j})$  & 
$\textrm{tr}(\bbR_{h_j}[\bbR_{y}(t-1)]_{11,j})+\sigma_{\epsilon_j}^2$ \\
\hline
\bfseries Global & $J^{-1}\textrm{tr}([\bbR_y(t)]_{11})$ & 
$J^{-1}\textrm{tr}(\bbR_{h}[\bbR_{y}(t-1)]_{11})$  &  
$J^{-1}\textrm{tr}(\bbR_{h}[\bbR_{y}(t-1)]_{11})+J^{-1}\sum_{j=1}^J\sigma_{
\epsilon_j}^2$\\
\hline
\end{tabular}
\end{table}

\begin{remark}
\normalfont Since the `average' system representation of $\bby(t)$ in \eqref{ystate_DRLS}
relies on an approximation that becomes increasingly accurate as $\lambda\to 1$ and $t\to\infty$,
so does the covariance recursion for $\bbR_y(t)$ derived in Section \ref{ssec:SS_Perf}.
For this reason, the scope of the MSE performance analysis of this paper pertains to 
the \textit{steady-state} behavior of the D-RLS algorithm.
\end{remark}

% % % % % % % % % % % % % % % % % % % % % % % % % % % % % % % % % % % % % % % %
%                         Section IV                                          %
% % % % % % % % % % % % % % % % % % % % % % % % % % % % % % % % % % % % % % % %

\section{Stability and Steady-State Performance Analysis}
\label{sec:Stability_Perf_Analysis}

In this section, stability and steady-state performance analyses are
conducted for the D-RLS algorithm developed in Section
\ref{ssec:AMA_DRLS}. Because recursions 
\eqref{ST_Vj_Update_AMA}-\eqref{ST_Sj_Update_AMA} are stochastic in nature, 
stability will be assessed both in the mean- and in the MSE-sense. The 
techniques presented here can be utilized with
minimal modifications to derive analogous results for the AD-MoM-based D-RLS 
algorithm in~\cite{Gonzalo_Yannis_GG_DRLS}.

% % % % % % % % % % % % % % % % % % % % % % % % % % % % % % % % % % % % % % % %
%                         Subsection IV-A                                     %
% % % % % % % % % % % % % % % % % % % % % % % % % % % % % % % % % % % % % % % %

\subsection{Mean Stability}
\label{ssec:DRLS_Mean_Stability}

Based on Lemma \ref{Lemma_1_Ch4}, it follows that D-RLS achieves consensus
in the mean sense on the parameter $\bbs_0$. 

\begin{proposition}\label{Proposition_3_Ch4}
Under (a1)-(a3) and for $0\ll\lambda<1$, D-RLS  achieves
consensus in the mean, i.e.,
\begin{equation*}
\lim_{t\to\infty}E[\bby_{1,j}(t)]=\mathbf{0}_p,{\quad}\forall\:j\in\mathcal{J}
\end{equation*}
provided the penalty coefficient is chosen such that
\begin{equation}\label{cbound}
0<c<\frac{4}{(1-\lambda)\lambda_{\max}(\bbR_{h}^{-1}(\bbL\otimes\bbI_p))}.
\end{equation}
\end{proposition}

\begin{IEEEproof}
Based on (a1)-(a3) and since the data is
zero-mean, one obtains after taking expectations on \eqref{ystate_DRLS}
that
$E[\bby(t)]=\mbox{bdiag}((1-\lambda)\bbR_{h}^{-1},\bbI_{Jp})\bbUpsilon
E[\bby(t-1)]$. The following lemma characterizes the spectrum of the
transition matrix
$\bbOmega:=\mbox{bdiag}((1-\lambda)\bbR_{h}^{-1},\bbI_{Jp})\bbUpsilon$; see 
Appendix B for a proof.

\begin{lemma}\label{Lemma_choose_c}
Regardless of the value of $c>0$, matrix
$\bbOmega:=\mbox{bdiag}((1-\lambda)\bbR_{h}^{-1},\bbI_{Jp})\bbUpsilon\in\mathbb
{R}^{2Jp\times
2Jp}$ has $p$ eigenvalues equal to one. Further, the left eigenvectors
associated with the unity eigenvalue have the structure $\bbv_i^{T}=\left[
\mathbf{0}_{1\times Jp}\:\: \bbq_i^{T}\right]$, where
$\bbq_i\in\textrm{nullspace}(\bbL_c)$ and $i=1,\ldots,p$. The remaining
eigenvalues are equal to zero, or else have modulus strictly smaller than
one provided $c$ satisfies the bound \eqref{cbound}.
\end{lemma}

Back to establishing the mean stability result, let $\{\bbu_i\}$ and
$\{\bbv_i^{T}\}$ respectively denote the collection of $p$ right and left
eigenvectors of $\bbOmega$ associated with the eigenvalue one. By virtue
of Lemma \ref{Lemma_choose_c} and provided $c$ satisfies the bound \eqref{cbound}, 
one has that $\lim_{t\to\infty}\bbOmega^{t}=\sum_{i=1}^{p}\bbu_i\bbv_i^{T}$; hence,
\begin{align}
\lim_{t\to\infty}E[\bby(t)]=&\left(\sum_{i=1}^{p}\bbu_i\bbv_i^{T}\right)
\bby(t_0)=\left(\sum_{i=1}^{p}\bbu_i\bbv_i^{T}\right)\textrm{bdiag}(\bbI_{Jp},
\bbL_c)\bby'(t_0)\nonumber\\
=&\left(\sum_{i=1}^{p}\bbu_i\left[\mathbf{0}_{1\times 
Jp}\:\:\bbq_i^{T}\bbL_c\right]\right)\bby'(t_0)
=\mathbf{0}_{2Jp}.\nonumber
\end{align}
In obtaining the second equality, the structure for $\bby(t_0)$
that is given in Lemma \ref{Lemma_1_Ch4} was used. The last equality follows from
the fact that $\bbq_i\in\textrm{nullspace}(\bbL_c)$ as per Lemma
\ref{Lemma_choose_c}, thus completing the proof.
\end{IEEEproof}

Before wrapping up this section, a comment is due on the sufficient
condition \eqref{cbound}. When performing distributed estimation under
$0\ll\lambda<1$, the condition is actually not restrictive at all since a
$1-\lambda$ factor is present in the denominator. When $\lambda$ is close
to one, any practical choice of $c>0$ will result in asymptotically
unbiased sensor estimates. Also note that the bound depends on the WSN 
topology, through the scaled graph Laplacian matrix $\bbL_c$.

% % % % % % % % % % % % % % % % % % % % % % % % % % % % % % % % % % % % % % % %
%                         Subsection IV-B                                     %
% % % % % % % % % % % % % % % % % % % % % % % % % % % % % % % % % % % % % % % %

\subsection{MSE Stability and Steady-State Performance }
\label{ssec:SS_Perf}

In order to assess the steady-state MSE performance
of the D-RLS algorithm, we will evaluate the figures of merit introduced
in Section \ref{ssec:fig_merit}. The limiting values of both the
local (per sensor) and global (network-wide) MSE, excess mean-square error
(EMSE), and mean-square deviation (MSD), will be assessed. To this end, it
suffices to derive a closed-form expression for the global estimation
error covariance matrix $\bbR_{y_1}(t):=E[\bby_1(t)\bby_1^{T}(t)]$, as
already argued in Section \ref{ssec:fig_merit}.

The next result provides an equivalent representation of the approximate
D-RLS global recursion \eqref{ystate_DRLS}, that is more suitable for the
recursive evaluation of $\bbR_{y_1}(t)$. First, introduce the
$p(\textstyle\sum_{j=1}^J|\calN_j|)\times 1$  vector
\begin{equation}\label{noisevector_Ch4}
\bbeta(t):=\left[\{(\bbeta_{j'}^{1}(t))^T\}_{j'\in\calN_{1}}\ldots
\{(\bbeta_{j'}^{J}(t))^T\}_{j'\in\calN_{J}}\right]^T
\end{equation}
which comprises the receiver noise terms corrupting transmissions of local
estimates across the whole network at time instant $t$, and define
$\bbR_{\bbeta}:=E[\bbeta(t)\bbeta^{T}(t)]$. For notational convenience,
let $\bbR_{h,\lambda}^{-1}:=(1-\lambda)\bbR_{h}^{-1}$.

\begin{lemma}\label{Lemma_2_Ch4}
Under the assumptions of Lemma \ref{Lemma_1_Ch4}, the global state
$\bby(t)$ in \eqref{ystate_DRLS} can be equivalently written as
\begin{equation}\label{D_RLS_System_Rec_2}
\bby(t+1)=\textrm{bdiag}(\bbI_{Jp},\bbL_c)
\bbz(t+1)+\left[\begin{array}{c} \bbR_{h,\lambda}^{-1}\\
\mathbf{0}_{Jp\times Jp}
\end{array}\right]\bar{\bbeta}(t)
+\left[\begin{array}{c} \bbR_{h,\lambda}^{-1}(\bbP_{\alpha}-\bbP_{\beta})\\
\bbP_{\beta}-\bbP_{\alpha}
\end{array}\right]{\bbeta}(t).
\end{equation}
The inner state $\mathbf{z}(t):=[\bbz_1^T(t)\;\bbz_2^T(t)]^T$ is
arbitrarily initialized at time $t_0$, and updated according to
\begin{equation}\label{zstate_Ch4}
\bbz(t+1)=\bbPsi\bbz(t)+\bbPsi\left[\begin{array}{c} \bbR_{h,\lambda}^{-1}(\bbP_{\alpha}-\bbP_{\beta})\\
\bbC
\end{array}\right]\bbeta(t-1)+\bbPsi\left[\begin{array}{c} \bbR_{h,\lambda}^{-1}\\
\mathbf{0}_{Jp\times Jp}
\end{array}\right]\bar{\bbeta}(t-1)
+\left[\begin{array}{c} \bbR_{h,\lambda}^{-1}\\
\mathbf{0}_{Jp\times Jp}
\end{array}\right]\bbepsilon(t+1)
\end{equation}
where the $2Jp\times 2Jp$ transition matrix $\bbPsi$ consists of the blocks
$[\bbPsi]_{11}=[\bbPsi]_{12}=-\bbR_{h,\lambda}^{-1}\bbL_c$ and
$[\bbPsi]_{21}=[\bbPsi]_{22}=\bbL_c\bbL_c^{\dagger}$. Matrix $\bbC$ is
chosen such that $\bbL_c\bbC=\bbP_{\beta}-\bbP_{\alpha}$, where the
structure of the time-invariant matrices $\bbP_{\alpha}$ and
$\bbP_{\beta}$ is given in Appendix E.
\end{lemma}
\begin{proof}
See Appendix C.
\end{proof}

The desired state $\bby(t)$ is obtained as a rank-deficient
linear transformation of the inner state $\bbz(t)$, plus a stochastic
offset due to the presence of communication noise. A linear, time-invariant,
first-order difference equation describes the dynamics of $\bbz(t)$, and
hence of $\bby(t)$, via the algebraic transformation in
\eqref{D_RLS_System_Rec_2}. The time-invariant nature of the
transition matrix $\bbPsi$ is due to the approximations $\bbPhi_j^{-1}(t)\approx
\bbR_{h,\lambda}^{-1}$, $j\in\calJ$, particularly accurate for large enough $t>t_0$.
Examination of \eqref{zstate_Ch4} reveals that the evolution of $\bbz(t)$ 
is driven by three stochastic input processes: i)
communication noise $\bbeta(t-1)$ affecting the transmission of local
estimates; ii) communication noise $\bar{\bbeta}(t-1)$ contaminating the
Lagrange multipliers; and iii) observation noise within $\bbepsilon(t+1)$.

Focusing now on the calculation of $\bbR_{y_1}(t)=[\bbR_{y}(t)]_{11}$
based on Lemma \ref{Lemma_2_Ch4}, observe from the upper
$Jp\times 1$ block of $\bby(t+1)$ in \eqref{D_RLS_System_Rec_2} that
$\bby_1(t+1)=\bbz_1(t+1)+\bbR_{h,\lambda}^{-1}[\bar{\bbeta}(t)+
(\bbP_{\alpha}-\bbP_\beta)\bbeta(t)]$. Under (a3), $\bbz_1(t+1)$ is
independent of the zero-mean $\{\bar{\bbeta}(t),\bbeta(t)\}$; hence,
\begin{equation}\label{yblock_covariance_Ch4}
\bbR_{y_1}(t)=\bbR_{z_1}(t)+\bbR_{h,\lambda}^{-1}\left[\bbR_{\bar{\bbeta}}+
(\bbP_{\alpha}-\bbP_\beta)\bbR_{\bbeta}(\bbP_{\alpha}-\bbP_\beta)^{T}\right]\bbR_{h,\lambda}^{-1}
\end{equation}
which prompts one to obtain $\bbR_{z}(t):=E[\bbz(t)\bbz^{T}(t)]$. Specifically,
 the goal is to 
extract its upper-left $Jp\times Jp$ matrix block
$[\bbR_{z}(t)]_{11}=\bbR_{z_1}(t)$. To this end, define the vectors
\begin{equation}\label{noiselambdavectors}
\bar{\bbeta}_{\lambda}(t):=\left[\begin{array}{c} \bbR_{h,\lambda}^{-1}\\
\mathbf{0}_{Jp\times Jp}
\end{array}\right]\bar{\bbeta}(t),{\quad}
\bbeta_{\lambda}(t):=\left[\begin{array}{c} \bbR_{h,\lambda}^{-1}(\bbP_{\alpha}-\bbP_{\beta})\\
\bbC
\end{array}\right]\bbeta(t)
\end{equation}
whose respective covariance matrices
$\bbR_{\bar{\bbeta}_{\lambda}}:=E[\bar{\bbeta}_{\lambda}(t)\bar{\bbeta}_{\lambda}^{T}(t)]$
and
$\bbR_{\bbeta_{\lambda}}:=E[\bbeta_{\lambda}(t)\bbeta_{\lambda}^{T}(t)]$
have a structure detailed in Appendix E. Also recall that
$\bbepsilon(t)$ depends on the entire history of regressors up to time
instant $t$. Starting from \eqref{zstate_Ch4} and capitalizing on 
(a2)-(a3), it is straightforward to
obtain a first-order matrix recursion to update $\bbR_{z}(t)$ as
\begin{align}
\bbR_{z}(t)={}&\bbPsi\bbR_{z}(t-1)\bbPsi^{T}+\bbPsi\bbR_{\bar{\bbeta}_{\lambda}}\bbPsi^{T}
+\bbPsi\bbR_{\bbeta_{\lambda}}\bbPsi^{T}+\left[\begin{array}{c} \bbR_{h,\lambda}^{-1}\\
\mathbf{0}_{Jp\times Jp}
\end{array}\right]\bbR_{\bbepsilon}(t)\left[\begin{array}{c} \bbR_{h,\lambda}^{-1}\\
\mathbf{0}_{Jp\times Jp}
\end{array}\right]^{T}\nonumber\\
&+\bbPsi\bbR_{z\bbepsilon}(t)\left[\begin{array}{c} \bbR_{h,\lambda}^{-1}\\
\mathbf{0}_{Jp\times Jp}
\end{array}\right]^T+\left(\bbPsi\bbR_{z\bbepsilon}(t)\left[\begin{array}{c} 
\bbR_{h,\lambda}^{-1}\\
\mathbf{0}_{Jp\times Jp}
\end{array}\right]^T\right)^{T}\label{z_covariance_rec}\\
:={}&\bbPsi\bbR_{z}(t-1)\bbPsi^{T}+\bbR_{\bbnu}(t)\label{z_covariance_simpler}
\end{align}
where the cross-correlation matrix
$\bbR_{z\bbepsilon}(t):=E[\bbz(t-1)\bbepsilon^{T}(t)]$ is recursively
updated as (cf. Appendix D)
\begin{equation}\label{cross_recursion}
\bbR_{z\bbepsilon}(t)=\lambda\bbPsi\bbR_{z\bbepsilon}(t-1)+\lambda\left[\begin{array}{c} \bbR_{h,\lambda}^{-1}\\
\mathbf{0}_{Jp\times Jp}
\end{array}\right]\bbR_{\bbepsilon}(t-1).
\end{equation}
For notational brevity in what follows, $\bbR_{\bbnu}(t)$ in
\eqref{z_covariance_simpler} denotes all the covariance forcing terms in
the right-hand side of \eqref{z_covariance_rec}. The main result of
this section pertains to MSE stability of the D-RLS algorithm, and
provides a checkable sufficient condition under which the global error
covariance matrix $\bbR_{y_1}(t)$ has bounded entries as $t\to\infty$. Recall 
that a matrix is termed stable, when all its eigenvalues lie strictly inside
 the unit circle.

\begin{proposition}\label{Proposition_4_Ch4}
Under (a1)-(a3) and for $0\ll\lambda<1$, D-RLS is MSE
stable, i.e., $\lim_{t\to\infty}\bbR_{y_1}(t)$ has bounded entries,
provided that $c>0$ is chosen so that $\bbPsi$ is a stable matrix.
\end{proposition}
\begin{proof}
First observe that because $\lambda\in(0,1)$, it holds that
\begin{align}
\lim_{t\to\infty}\bbR_{\bbepsilon}(t)=&\lim_{t\to\infty}\left(\frac{1-
\lambda^{2(t+1)}}{1-\lambda^2}\right)
\textrm{bdiag}(\bbR_{h_1}\sigma_{\epsilon_1}^2,\ldots,\bbR_{h_J}\sigma_{
\epsilon_J}^2)\nonumber\\
=&\left(\frac{1}{1-\lambda^2}\right)
\textrm{bdiag}(\bbR_{h_1}\sigma_{\epsilon_1}^2,\ldots,\bbR_{h_J}\sigma_{
\epsilon_J}^2)=:\bbR_{\bbepsilon}(\infty).
\label{R_Epsilon_infty}
\end{align}
If $c>0$ is selected such that $\bbPsi$ is a stable matrix, then clearly
$\lambda\bbPsi$ is also stable, and hence the matrix recursion
\eqref{cross_recursion} converges to the bounded limit
\begin{equation}
\bbR_{z\bbepsilon}(\infty)=\left(\bbI_{2Jp}-\lambda\bbPsi\right)^{-1}\left[
\begin{array}{c} \lambda\bbR_{h,\lambda}^{-1}\\
\mathbf{0}_{Jp\times Jp}
\end{array}\right]\bbR_{\bbepsilon}(\infty).\label{R_HEpsilon_infty}
\end{equation}
Based on the previous arguments, it follows that the forcing matrix
$\bbR_{\bbnu}(t)$ in \eqref{z_covariance_rec} will also attain a bounded
limit as $t\to\infty$, denoted as $\bbR_{\bbnu}(\infty)$. Next, we show
that $\lim_{t\to\infty}\bbR_{z}(t)$ has bounded entries by studying its
equivalent vectorized dynamical system. Upon vectorizing
\eqref{z_covariance_simpler}, it follows that
\begin{align}
\textrm{vec}[\bbR_{z}(t)]=&\textrm{vec}[\bbPsi\bbR_{z}(t-1)\bbPsi^{T}]+
\textrm{vec}[\bbR_{\bbnu}(t)]\nonumber\\
=&\left(\bbPsi\otimes\bbPsi\right)\textrm{vec}[\bbR_{z}(t-1)]+
\textrm{vec}[\bbR_{\bbnu}(t)]
\nonumber
\end{align}
where in obtaining the last equality we used the property
$\textrm{vec}[\bbR\bbS\bbT]=\left(\bbT^{T}\otimes\bbR\right) 
\textrm{vec}[\bbS]$. Because
the eigenvalues of $\bbPsi\otimes\bbPsi$ are the pairwise products of
those of $\bbPsi$, stability of $\bbPsi$ implies stability of the
Kronecker product. As a result, the vectorized recursion will converge to
the limit
\begin{equation}\label{veclimit}
\textrm{vec}[\bbR_{z}(\infty)]=\left(\bbI_{(2Jp)^2}-\bbPsi\otimes\bbPsi\right)^
{-1}\textrm{vec}[\bbR_{\bbnu}(\infty)]
\end{equation}
which of course implies that
$\lim_{t\to\infty}\bbR_{z}(t)=\bbR_{z}(\infty)$ has bounded entries. From
\eqref{yblock_covariance_Ch4}, the same holds true for $\bbR_{y_1}(t)$,
and the proof is completed.
\end{proof}

Proposition \ref{Proposition_4_Ch4} asserts that the AMA-based D-RLS
algorithm is stable in the MSE-sense, even when the WSN links are
challenged by additive noise. While most distributed adaptive estimation
works have only looked at ideal inter-sensor links, others 
have adopted diminishing step-sizes to mitigate the
undesirable effects of communication noise~\cite{Kar_Consensus_Noise_TSP,Hatano_Consensus_Noise_2007}.
This approach however, limits their applicability to stationary environments. 
Remarkably, the AMA-based D-RLS algorithm exhibits robustness to noise when
using a constant
step-size $c$, a feature that has also been observed for AD-MoM
related distributed iterations in 
e.g.,~\cite{Yannis_Ale_GG_PartI,Yannis_Gonzalo_GG_DLMS}, and~\cite{Gonzalo_Yannis_GG_DRLS}.

As a byproduct, the proof of Proposition \ref{Proposition_4_Ch4} also
provides part of the recipe towards evaluating the steady-state MSE
performance of the D-RLS algorithm. Indeed, by plugging
\eqref{R_Epsilon_infty} and \eqref{R_HEpsilon_infty} into
\eqref{z_covariance_rec} one obtains the steady-state covariance matrix
$\bbR_{\bbnu}(\infty)$. It is then possible to evaluate
$\bbR_{z}(\infty)$, by reshaping the vectorized identity \eqref{veclimit}.
Matrix $\bbR_{z_1}(\infty)$ can be extracted from the upper-left $Jp\times
Jp$ matrix block of $\bbR_{z}(\infty)$, and the desired global error
covariance matrix $\bbR_{y_1}(\infty)=[\bbR_{y}(\infty)]_{11}$ becomes 
available via
\eqref{yblock_covariance_Ch4}. Closed-form evaluation of the
MSE$(\infty)$, EMSE$(\infty)$ and MSD$(\infty)$ for every sensor
$j\in\calJ$ is now possible given $\bbR_{y_1}(\infty)$, by resorting to
the formulae in Table \ref{table:figures_of_merit}.

Before closing this section, an alternative notion of stochastic stability that 
readily follows from Proposition \ref{Proposition_4_Ch4} is established here. 
Specifically, it is
possible to show that under the independence setting assumptions (a1)-(a3)
considered so far, the global
error norm $\|\bby_1(t)\|$ remains most of the time within a finite interval,
i.e., errors are weakly stochastic bounded
(WSB)~\cite{Solo_Stability_LMS},~\cite[pg. 110]{Solo_Adaptive_Book}. This
WSB stability guarantees that for any $\theta>0$, there exists a $\zeta>0$
such that $\textrm{Pr}[\|\bby_1(t)\|<\zeta]=1-\theta$ uniformly in time. 

\begin{corollary}\label{corollary_1}
Under (a1)-(a3) and for $0\ll\lambda<1$, if $c>0$ is chosen so that 
$\bbPsi$ is a stable matrix, then the D-RLS algorithm yields estimation errors 
which are WSB;
i.e., $\lim_{\zeta\rightarrow\infty}\sup_{t\geq t_0}\textrm{\emph{Pr}}[\|\bby_1(t)\|\geq\zeta]=0.$
\end{corollary}
\begin{IEEEproof}
Chebyshev's inequality implies that
\begin{equation}\label{Chebyshev_application}
\textrm{Pr}[\|\bby_1(t)\|\geq\zeta]\leq \frac{E[\|\bby_1(t)\|^{2}]}{\zeta^2}
=\frac{\textrm{tr}([\bbR_y(t)]_{11})}{\zeta^2}.
\end{equation}
From Proposition \ref{Proposition_4_Ch4}, $\lim_{t\to\infty}[\bbR_y(t)]_{11}$
has bounded entries, implying that $\sup_{t\geq
t_0}\textrm{tr}([\bbR_y(t)]_{11})<\infty$. Taking the limit as
$\zeta\to\infty$, while relying on the bound in
\eqref{Chebyshev_application} which holds for all values of $t\geq t_0$, yields the desired result.
\end{IEEEproof}

In
words, Corollary \ref{corollary_1} ensures that with overwhelming probability, 
local sensor estimates remain inside a ball with finite radius, centered
 at $\bbs_0$.
It is certainly a weak notion of stability, many times the only one that
can be asserted when the presence of, e.g., time-correlated
data, renders variance calculations impossible; see
also~\cite{Yannis_Gonzalo_GG_DLMS},~\cite{Solo_Stability_LMS}. In this case
where  stronger assumptions are invoked, WSB follows immediately once
MSE-sense stability is established.
Nevertheless, it is an important practical notion as it ensures -- on a
per-realization basis -- that 
estimation errors have no probability mass escaping to infinity. In particular, D-RLS estimation 
errors are shown WSB in the presence of communication noise; a property
not enjoyed by other distributed iterations for e.g., consenting 
on averages~\cite{Xiao_Fast_Iterations_2004}.

% % % % % % % % % % % % % % % % % % % % % % % % % % % % % % % % % % % % % % % %
%                         Section V                                           %
% % % % % % % % % % % % % % % % % % % % % % % % % % % % % % % % % % % % % % % %

\section{Numerical Tests}\label{sec:sims_Perf}

Computer simulations are carried out here to corroborate the
analytical results of Section \ref{ssec:SS_Perf}. Even though based on
simplifying assumptions and approximations, the usefulness of the analysis
is justified since the predicted steady-state MSE figures of merit
accurately match the empirical D-RLS limiting values. In accordance
with the adaptive filtering folklore, when $\lambda\to 1$ 
the upshot of the analysis under the
independence setting assumptions is shown to extend accurately
to the pragmatic scenario whereby sensors acquire time-correlated
data. For $J=15$ sensors,
a connected ad hoc WSN is generated as a realization of the random
geometric graph model on the unit-square, with communication range
$r=0.3$~\cite{Gupta_Kumar_Capacity_WN}. To model non-ideal inter-sensor links, 
additive white Gaussian noise (AWGN) with
variance $\sigma_\eta^2=10^{-1}$ is added at the receiving end. The WSN 
used for the experiments is depicted in Fig. \ref{wsn}.

With $p=4$ and $\bbs_0=\mathbf{1}_p$, observations obey a linear model
[cf. (a1)] with sensing WGN of spatial variance profile
$\sigma_{\epsilon_j}^2=10^{-3}\alpha_j$, where
$\alpha_j\sim\mathcal{U}[0,1]$ (uniform distribution) and i.i.d.. The
regression vectors $\bbh_j(t):=[h_j(t)\ldots h_j(t-p+1)]^{T}$ have a shift
structure, and entries which evolve according to first-order stable autoregressive
processes
$h_j(t)=(1-\rho)\beta_jh_j(t-1)+\sqrt{\rho}\omega_j(t)$ for all
$j\in\calJ$. We choose $\rho=5\times 10^{-1}$, the
$\beta_j\sim\mathcal{U}[0,1]$ i.i.d. in space, and the driving white noise
$\omega_j(t)\sim\mathcal{U}[-\sqrt{3}\sigma_{\omega_j},\sqrt{3}\sigma_{\omega_j}]$
with spatial variance profile given by $\sigma_{\omega_j}^{2}=2\gamma_j$
with $\gamma_j\sim\mathcal{U}[0,1]$ and i.i.d.. Observe that the data is
temporally-correlated, implying that (a2) does not hold here.

For all experimental performance curves obtained by running the
algorithms, the ensemble averages are approximated by sample averaging
$200$ runs of the experiment.

First, with $\lambda=0.95$, $c=0.1$ and $\delta=100$ for the AMA-based
D-RLS algorithm, Fig. \ref{global_perf_DRLS} depicts the network
performance through the evolution of the $\textrm{EMSE}(t)$ and
$\textrm{MSD}(t)$ figures of merit. Both noisy and ideal links are
considered. The steady-state limiting
values found in Section \ref{ssec:SS_Perf} are extremely accurate, even
though the simulated data does not adhere to (a2), and the results are
based on simplifying approximations. As intuitively expected and
analytically corroborated via the noise-related additive terms in
\eqref{yblock_covariance_Ch4} and \eqref{z_covariance_rec}, the
performance penalty due to non-ideal links is also apparent.

We also utilize the analytical results developed throughout this paper to
contrast the per sensor performance of D-RLS and the D-LMS algorithm
in~\cite{Gonzalo_Yannis_GG_DLMS_Perf}. In particular, the parameters chosen for D-LMS are
$\mu=5\times 10^{-3}$ and $c=1$. Fig. \ref{local_perf_DRLS} shows the
values of the $\textrm{EMSE}_j(\infty)$ and $\textrm{MSD}_j(\infty)$ for
all $j\in\calJ$. As expected, the second-order D-RLS scheme attains
improved steady-state performance uniformly across all sensors in the
simulated WSN. In this particular simulated test, gains as high as $5$dB
in estimation error can be achieved at the price of increasing
computational burden per sensor, from $\mathcal{O}(p)$ to $\mathcal{O}(p^2)$ per iteration.

% % % % % % % % % % % % % % % % % % % % % % % % % % % % % % % % % % % % % % % %
%                         Section VI                                          %
% % % % % % % % % % % % % % % % % % % % % % % % % % % % % % % % % % % % % % % %

\section{Concluding Summary and Future Work}\label{sec:conc} 
A distributed RLS-like algorithm is developed in this paper, which is
capable of performing adaptive
estimation and tracking using WSNs in which sensors cooperate
with single-hop neighbors. The WSNs considered here are quite general 
since they do not necessarily possess a Hamiltonian cycle, while the 
inter-sensor links are challenged by communication noise. Distributed
iterations are derived after: i)
reformulating in a separable way the exponentially weighed
least-squares (EWLS) cost involved in the classical RLS algorithm; and
ii) applying the AMA to minimize this separable cost in
a distributed fashion. The AMA is especially well-suited
to capitalize on the strict convexity of the EWLS cost, and thus 
offer significant reductions in computational complexity per sensor,
when compared to existing alternatives.
This way, salient features of the classical RLS algorithm are shown to carry over
to a distributed WSN setting, namely reduced-complexity estimation 
when a state and/or data model is not available and fast convergence 
rates are at a premium. 

An additional contribution of this paper pertains to a detailed steady-state
MSE performance analysis, that relies on an `averaged'
error-form system representation of D-RLS.
The theory is developed under some simplifying
approximations, and resorting to the independence setting assumptions. 
This way, it is possible to obtain accurate closed-form expressions 
for both the per sensor and network-wide
relevant performance metrics as $t\to\infty$. Sufficient conditions under which
the D-RLS algorithm is stable in the mean- and MSE-sense are provided as well. 
As a corollary, the D-RLS estimation errors are also shown to  
remain within a finite interval with high probability, even when
the inter-sensor links are challenged by additive noise.
Numerical simulations demonstrated that the analytical findings of this paper
extend accurately to a more realistic WSN setting, whereby sensors acquire temporally 
correlated sensor data.

Regarding the performance of the D-RLS algorithm, there are still several 
interesting directions to pursue as future work. 
First, it would be nice to establish a stochastic \textit{trajectory locking} result 
which formally shows that as $\lambda\to 1$, the D-RLS estimation error trajectories
closely follow the ones of its time-invariant `averaged'
system companion.
Second, the steady-state MSE performance analysis was
carried out when $0\ll\lambda<1$. For the infinite
memory case in which $\lambda=1$, numerical simulations indicate that
D-RLS provides mean-square sense-consistent estimates, even in the
presence of communication noise. By formally establishing this property, D-RLS
becomes an even more appealing alternative for distributed parameter
estimation in stationary environments. While the approximations used in
this paper are no longer valid when $\lambda=1$, for Gaussian
i.i.d. regressors matrix $\bbPhi^{-1}(t)$ is Wishart distributed with 
known moments. Under these assumptions, consistency analysis is a subject of 
ongoing investigation.

% % % % % % % % % % % % % % % % % % % % % % % % % % % % % % % % % % % % % % % %
%                         Appendices                                          %
% % % % % % % % % % % % % % % % % % % % % % % % % % % % % % % % % % % % % % % %

{\Large\appendix}

% % % % % % % % % % % % % % % % % % % % % % % % % % % % % % % % % % % % % % % %
%                         Appendix A                                          %
% % % % % % % % % % % % % % % % % % % % % % % % % % % % % % % % % % % % % % % %

\noindent\normalsize \emph{\textbf{A. Proof of Lemma}
\ref{Lemma_1_Ch4}}: Let $t_0$ be chosen large enough to ensure 
that 
\begin{equation*}
\lim_{t\to t_0}\mathbf{\Phi}_j(t)=\lim_{t\to 
t_0}\sum_{\tau=0}^{t}\lambda^{t-\tau}\mathbf{h}_j(\tau)\mathbf{h}_j^T(\tau)+J
^{-1}\lambda^{t}\mathbf{\Phi}_0\approx \frac{\bbR_{h_j}}{1-\lambda},\quad 
j\in\calJ.
\end{equation*}
For $t>t_0$, consider replacing $\mathbf{\Phi}_j^{-1}(t)$ in 
\eqref{ST_Sj_Update_AMA} with
the approximation $(1-\lambda)\bbR_{h_j}^{-1}$ for its expected value, to 
arrive at the `average' D-RLS system recursions
\begin{align}
\mathbf{v}_j^{j'}(t)={}&\mathbf{v}_{j}^{j'}(t-1)
+\frac{c}{2}\left[\mathbf{s}_{j}(t)-({\mathbf{s}}_{j'}(t)+\bbeta_j^{j'}(t))
\right],
{\quad}j'\in\calN_j\label{Vj_Update_averaged}\\
\bbs_j(t+1)={}&(1-\lambda)\bbR_{h_j}^{-1}\bbpsi_j(t+1)
-\frac{1}{2}(1-\lambda)\bbR_{h_j}^{-1}\sum_{j'\in\calN_j}\left[\bbv_{j}^{j'}(t)
-(\bbv_{j'}^j(t)+\bar{\bbeta}_{j}^{j'}(t))\right]\label{Sj_Update_averaged}
\end{align}

After summing $(\bbv_{j}^{j'}(t)-\bbv^{j}_{j'}(t))/2$ over $j'\in\calN_j$, it
follows from \eqref{Vj_Update_averaged} that for all $j\in\calJ$
\begin{align}
\hspace{-0.3cm}\bby_{2,j}(t+1):=&\:\frac{1}{2}\sum_{j'\in\calN_j}(\bbv_{j}^{j'}
(t)-\bbv^{j}_{j
'}(t))=\bby_{2,j}(t)+\frac{c}{2}\sum_{j'\in\calN_j}(\bbs_j(t)-\bbs_{j'}(t))-
\frac{c}{4}\sum_{j'\in\calN_j}(\bbeta_{j}^{j'}(t)-
\bbeta^{j}_{j'}(t))\label{y2recursionfirstsetp}\\
=&\: 
\bby_{2,j}(t)+\frac{c}{2}\sum_{j'\in\calN_j}(\bby_{1,j}(t)-\bby_{1,j'}(t))-
\bbeta_j^\alpha(t)+\bbeta_j^\beta(t),\label{y2recursion}
\end{align}
where the last equality was obtained after adding and subtracting
$c|\calN_j|\bbs_0$ from the right-hand side of 
\eqref{y2recursionfirstsetp}, and
relying on the definitions in \eqref{noises_alpha_beta}. Next, starting
from \eqref{Sj_Update_averaged} and upon: i) using (a1) to eliminate
\begin{equation*}
\bbpsi_j(t+1)=\sum_{\tau=0}^{t+1}\lambda^{t+1-\tau}\bbh_j(\tau) 
\bbh_j^T(\tau)\bbs_0+ 
\sum_{\tau=0}^{t+1}\lambda^{t+1-\tau}\bbh_j(\tau)\epsilon_j(\tau)\approx
\frac{\bbR_{h_j}}{1-\lambda}\bbs_0+\sum_{\tau=0}^{t+1}\lambda^{t+1-\tau}\bbh_j(
\tau)\epsilon_j(\tau)
\end{equation*}
from
\eqref{Sj_Update_averaged};
ii) recognizing $\bby_{2,j}(t+1)$ in the right-hand side of 
\eqref{Sj_Update_averaged} and
substituting it with \eqref{y2recursion}; and iii) replacing the sums of noise 
vectors with the
quantities defined in \eqref{baretajt_Ch4} and \eqref{noises_alpha_beta}; one 
arrives at
\begin{align}\label{y1recursion}
\bby_{1,j}(t+1)
=&\:(1-\lambda)\bbR_{h_j}^{-1}\left[-\frac{c}{2}\sum_{j'\in\calN_j}(\bby_{1,j}(
t)-\bby_{1,j'}(t))
-\bby_{2,j}(t)\right]
\nonumber\\
&+(1-\lambda)\bbR_{h_j}^{-1}\left[\sum_{\tau=0}^{t+1}\lambda^{t+1-\tau}\bbh_j(
\tau)\epsilon_j(\tau)+\bbeta_j^\alpha(t)-\bbeta_j^\beta(t)+\bar{\bbeta}_j(t)
\right].
\end{align}

What remains to be shown is that after stacking the recursions
\eqref{y1recursion} and \eqref{y2recursion} for $j=1,\ldots,J$ to form the
one for $\bby(t+1)$, we can obtain the compact representation in
\eqref{ystate_DRLS}. Examining \eqref{y2recursion} and 
\eqref{y1recursion}, it is apparent that a common matrix factor 
$\textrm{bdiag}((1-\lambda)\bbR_{h_j}^{-1},\bbI_{Jp})$ can be pulled out to 
symplify the expression for $\bby(t+1)$.  Consider first the forcing terms in 
\eqref{ystate_DRLS}.
Stacking the channel noise terms from \eqref{y1recursion} and
\eqref{y2recursion}, readily yields the last three terms inside the curly 
brackets in \eqref{ystate_DRLS}. Likewise, stacking the terms
$\sum_{\tau=0}^{t+1}\lambda^{t+1-\tau}\bbh_j(
\tau)\epsilon_j(\tau)$ for $j=1,\ldots,J$ yields the second term due to the 
observation noise; recall the definition of $\bbepsilon(t+1)$. This term as 
well as the  vectors 
$\bar{\bbeta}_j(t)$ are not present in
\eqref{y2recursion}, which explains the zero vector at the lower part of
the second and third terms inside the curly brackets of \eqref{ystate_DRLS}.

To specify the structure of the transition matrix $\bbUpsilon$,
note that the first term on the right-hand side of \eqref{y2recursion} explains
 why
$[\bbUpsilon]_{22}=\bbI_{Jp}$. Similarly, the second term inside
the first square brackets in \eqref{y1recursion} explains why
$[\bbUpsilon]_{12}=-\bbI_{Jp}$. Next, it follows readily that
upon stacking the terms
$(c/2)\sum_{j'\in\calN_j}(\bby_{1,j}(t)-\bby_{1,j'}(t))$, which correspond to
a scaled Laplacian-based combination of $p\times 1$ vectors, one obtains
$[(c/2)\bbL\otimes \mathbf{I}_{p}]\bby_1(t)=\bbL_c\bby_1(t)$. This justifies
why $[\bbUpsilon]_{11}=-[\bbUpsilon]_{21}=-\bbL_{c}$. 

A comment is due regarding the initialization for $t=t_0$. Although the vectors
 $\{\bby_{1,j}(t_0)\}_{j=1}^{J}$ are decoupled so that
$\bby_1(t_0)$ can be chosen arbitrarily, this is not the case for
$\{\bby_{2,j}(t_0)\}_{j=1}^{J}$ which are coupled and satisfy
\begin{equation}\label{multicoupling}
\sum_{j=1}^{J}\bby_{2,j}(t)=\sum_{j=1}^{J}\sum_{j'\in\calN_j}
(\bbv_{j}^{j'}(t-1)-\bbv^{j}_{j'}(t-1))=\mathbf{0}_{p},{\quad}\forall\; t\geq 
0.
\end{equation}
The coupling across $\{\bby_{2,j}(t)\}_{j=1}^{J}$ dictates 
$\bby_2(t_0)$ to be chosen in compliance with \eqref{multicoupling}, so
that the system \eqref{ystate_DRLS} is equivalent to \eqref{Vj_Update_averaged}
 and \eqref{Sj_Update_averaged} for all $t\geq t_0$.
Let $\bby_2(t_0)=\bbL_c\bby_2'(t_0)$, where $\bby_2'(t_0)$ is any vector in
$\mathbb{R}^{Jp}$. Then, it is not difficult to see that $\bby_2(t_0)$ satisfies the
 conservation law
\eqref{multicoupling}.
In conclusion, for arbitrary $\bby'(0)\in\mathbb{R}^{2Jp}$ the recursion
\eqref{ystate_DRLS} should be initialized as
$\bby(0)=\textrm{bdiag}(\bbI_{Jp},\bbL_c)\bby'(0)$, and the proof of Lemma
\ref{Lemma_1_Ch4} is completed.
\hfill$\blacksquare$

% % % % % % % % % % % % % % % % % % % % % % % % % % % % % % % % % % % % % % % %
%                         Appendix B                                          %
% % % % % % % % % % % % % % % % % % % % % % % % % % % % % % % % % % % % % % % %

\noindent\normalsize \emph{\textbf{B. Proof of Lemma}
\ref{Lemma_choose_c}}: Recall the structure of matrix $\bbUpsilon$ given in 
Lemma
\ref{Lemma_1_Ch4}. A vector $\bbv_i^{T}:=\left[\bbv_{1,i}^{T}\:\:
\bbv_{2,i}^{T}\right]$ with $\{\bbv_{j,i}\}_{j=1}^2\in\mathbb{R}^{Jp\times
1}$ is a left eigenvector of $\bbOmega$ associated to the eigenvalue one,
if and only if it solves the following linear system of equations
\begin{align}
-\bbv_{1,i}^{T}(1-\lambda)\bbR_h^{-1}\bbL_c+\bbv_{2,i}^{T}\bbL_c&=\bbv_{1,i}^{T}\nonumber\\
-\bbv_{1,i}^{T}(1-\lambda)\bbR_h^{-1}+\bbv_{2,i}^{T}&=\bbv_{2,i}^{T}\nonumber
\end{align}
The second equation can only be satisfied for
$\bbv_{1,i}=\mathbf{0}_{Jp}$, and upon substituting this value in the
first equation one obtains that $\bbv_{2,i}\in\textrm{nullspace}(\bbL_c)=
\textrm{nullspace}(\bbL\otimes \bbI_p)$ for all values of $c>0$. Under the
assumption of a connected ad hoc WSN, 
$\textrm{nullspace}(\bbL)=\textrm{span}(\mathbf{1}_{J})$ and hence
$\textrm{nullspace}(\bbL\otimes \bbI_p)$ is a $p$-dimensional subspace.

Following steps similar to those in~\cite[Appendix
H]{Yannis_Ale_GG_PartI}, it is possible to express the eigenvalues of
$\bbOmega$ that are different from one as the roots of a second-order
polynomial. Such a polynomial does not have an independent term, so that
some eigenvalues are zero. With respect to the rest of the eigenvalues, it
is possible to show that their magnitude is upper bounded by
$\lambda_{\max}(\bbI_{Jp}-(1-\lambda)\bbR_h^{-1}\bbL_c)$. 
%Note that
%%
%\begin{equation}\label{similarity}
%\bbR_h^{1/2}\left[(1-\lambda)\bbR_h^{-1}\bbL_c\right]\bbR_h^{-1/2}=
%\frac{c}{2}(1-\lambda)\bbR_h^{-1/2}(\bbL\otimes \bbI_p)\bbR_h^{-1/2}
%\end{equation}
%has the same eigenvalues as $(1-\lambda)\bbR_h^{-1}\bbL_c$ because these
%are invariant under similarity transformations. Focusing on the right hand
%side of \eqref{similarity}, from Sylvester's law of intertia~\cite[p.
%403]{Golub_Book} it follows that all eigenvalues of
%$(1-\lambda)\bbR_h^{-1}\bbL_c$ are real and nonnegative. 
Hence, it is
possible to select $c>0$ such that
$\lambda_{\max}(\bbI_{Jp}-(1-\lambda)\bbR_h^{-1}\bbL_c)<1$, or
equivalently $|1-(1-\lambda)\lambda_{\max}(\bbR_h^{-1}\bbL_c)|<1$, which is
the same as condition \eqref{cbound}.\hfill$\blacksquare$

% % % % % % % % % % % % % % % % % % % % % % % % % % % % % % % % % % % % % % % %
%                         Appendix C                                          %
% % % % % % % % % % % % % % % % % % % % % % % % % % % % % % % % % % % % % % % %

\noindent\normalsize \emph{\textbf{C. Proof of Lemma}
\ref{Proposition_4_Ch4}}: The goal is to establish 
the equivalence between
the dynamical systems in \eqref{ystate_DRLS} and \eqref{D_RLS_System_Rec_2} for
all $t\geq t_0$, when the inner state is arbitrarily initialized as
$\bbz(t_0)=\bby'(t_0)$. We will argue by induction. For $t=t_0$, it follows 
from
\eqref{zstate_Ch4} that
$\bbz(t_0+1)=\bbPsi\bby'(t_0)+[\bbR_{h,\lambda}^{-1}\;\mathbf{0}^{T}]^{T}
\bbepsilon(t_0+1)$, since (by convention) there is no communication noise 
for $t<t_0$.
Upon substituting $\bbz(t_0+1)$ into \eqref{D_RLS_System_Rec_2}, we find
\begin{equation}\label{ystate_initial}
\bby(t_0+1)=\textrm{bdiag}(\bbI_{Jp},\bbL_c)\bbPsi\bby'(t_0)+\left[
\begin{array}{c} \bbR_{h,\lambda}^{-1}\\
\mathbf{0}_{Jp\times Jp}
\end{array}\right](\bbepsilon(t_0+1)+\bar{\bbeta}(t_0))+\left[\begin{array}{c} 
\bbR_{h,\lambda}^{-1}(\bbP_{\alpha}-\bbP_{\beta})\\
\bbP_{\beta}-\bbP_{\alpha}
\end{array}\right]{\bbeta}(t_0).
\end{equation}
Note that: i)
$\textrm{bdiag}(\bbI_{Jp},\bbL_c)\bbPsi=\bbUpsilon\textrm{bdiag}(
\bbI_{Jp},\bbL_c)$; ii) $\bby(t_0)=\textrm{bdiag}(\bbI_{Jp},\bbL_c)\bby'(t_0)$
for the system in Lemma \ref{Lemma_1_Ch4}; and iii)
$\bbeta_{\alpha}(t)=\bbP_\alpha\bbeta(t)$, while
$\bbeta_{\beta}(t)=\bbP_\beta\bbeta(t)$ [cf. Appendix E]. Thus, the right-hand 
side of
\eqref{ystate_initial} is equal to the right-hand side of \eqref{ystate_DRLS} 
for $t=t_0$.

Suppose next that \eqref{D_RLS_System_Rec_2} and \eqref{zstate_Ch4} hold true
for $\bby(t)$ and $\bbz(t)$, with $t\geq t_0$. The same will be shown for 
$\bby(t+1)$ and
$\bbz(t+1)$. To this end, replace $\bby(t)$ with the right-hand side of
\eqref{D_RLS_System_Rec_2} evaluated at time $t$, into
\eqref{ystate_DRLS} to obtain
\begin{align}\label{ystate_induction}
\bby(t+1)=&\:\textrm{bdiag}(\bbR_{h,\lambda}^{-1},\bbI_{Jp})\left\{
\bbUpsilon\textrm{bdiag}(\bbI_{Jp},\bbL_c)\bbz(t)
+\bbUpsilon\left[\begin{array}{c} \bbR_{h,\lambda}^{-1}\\
\mathbf{0}_{Jp\times Jp}
\end{array}\right]\bar{\bbeta}(t-1)+\left[\begin{array}{c} \bbI_{Jp}\\
\mathbf{0}_{Jp\times Jp}
\end{array}\right]\bbepsilon(t+1)
\right.\nonumber\\
&\left.+\bbUpsilon\left[\begin{array}{c} \bbR_{h,\lambda}^{-1}  
(\bbP_{\alpha}-\bbP_{\beta})\\
\bbP_{\beta}-\bbP_{\alpha}
\end{array}\right]{\bbeta}(t-1)
+\left[\begin{array}{c} \bbI_{Jp}\\
\mathbf{0}_{Jp\times Jp}
\end{array}\right]\bar{\bbeta}(t)
+\left[\begin{array}{c} \bbI_{Jp}\\
-\bbI_{Jp}
\end{array}\right]\bbeta_\alpha(t)-
\left[\begin{array}{c} \bbI_{Jp}\\
-\bbI_{Jp}
\end{array}\right]\bbeta_\beta(t)\right\}\nonumber\\
=&\:\textrm{bdiag}(\bbI_{Jp},\bbL_c)\left(\bbPsi\bbz(t)
+\bbPsi\left[\begin{array}{c} \bbR_{h,\lambda}^{-1}\\
\mathbf{0}_{Jp\times Jp}
\end{array}\right]\bar{\bbeta}(t-1)
+\bbPsi\left[\begin{array}{c}  
\bbR_{h,\lambda}^{-1}(\bbP_{\alpha}-\bbP_{\beta})\\
\bbC
\end{array}\right]\bbeta(t-1)\right.\nonumber\\
&\left.+\left[
\begin{array}{c} \bbR_{h,\lambda}^{-1}\\
\mathbf{0}_{Jp\times Jp}
\end{array}\right]\bbepsilon(t+1)\right)+\left[\begin{array}{c} 
\bbR_{h,\lambda}^{-1}\\
\mathbf{0}_{Jp\times Jp}
\end{array}\right]\bar{\bbeta}(t)
+\left[\begin{array}{c} \bbR_{h,\lambda}^{-1}(\bbP_{\alpha}-\bbP_{\beta})\\
\bbP_{\beta}-\bbP_{\alpha}
\end{array}\right]{\bbeta}(t)
\end{align}
where in obtaining the last equality in \eqref{ystate_induction}, the following were 
used: i)
$\textrm{bdiag}(\bbI_{Jp},\bbL_c)\bbPsi=\bbUpsilon\textrm{bdiag}(
\bbI_{Jp},\bbL_c)$
; ii) the relationship between $\bbeta_{\alpha}(t), \bbeta_{\beta}(t)$ and
$\bbeta(t)$ given in Appendix E; and iii) the existence of a matrix $\bbC$ such
 that $\bbL_c
\bbC=\bbP_\beta-\bbP_\alpha$. This made possible to extract the common
factor $\textrm{bdiag}(\bbI_{Jp},\bbL_c)$ and deduce from
\eqref{ystate_induction} that $\bby(t+1)$ is given by
\eqref{D_RLS_System_Rec_2}, while $\bbz(t+1)$ is provided by
\eqref{zstate_Ch4}.

In order to complete the proof, one must show the existence of matrix
$\bbC$. To this end, via a simple evaluation one can check
that
$\textrm{nullspace}(\bbL_c)\subseteq\textrm{nullspace}(\bbP_{\beta}^T-\bbP_{
\alpha}^T)$,
and since $\bbL_c$ is symmetric, one has
$\textrm{nullspace}(\bbL_c)\bot\textrm{range}(\bbL_c)$. As
$\textrm{nullspace}(\bbP_{\beta}^T-\bbP_{\alpha}^T)\bot\textrm{range}(\bbP_{
\beta}-\bbP_{\alpha})$,
it follows that
$\textrm{range}(\bbP_{\beta}-\bbP_{\alpha})\subseteq\textrm{range}(\bbL_c)$,
which further implies that there exists  $\bbC$ such that
$\bbL_c\bbC=\bbP_{\beta}-\bbP_{\alpha}$.\hfill$\blacksquare$

% % % % % % % % % % % % % % % % % % % % % % % % % % % % % % % % % % % % % % % %
%                         Appendix D                                          %
% % % % % % % % % % % % % % % % % % % % % % % % % % % % % % % % % % % % % % % %

\noindent\normalsize \emph{\textbf{D. Derivation of \eqref{cross_recursion}}}: 
First observe that the noise supervector $\bbepsilon(t)$ obeys the first-order 
recursion
\begin{equation}\label{recursion_epsilon}
\bbepsilon(t):=\sum_{\tau=0}^{t}\lambda^{t-\tau}[\bbh_1^T(\tau)\epsilon_1(\tau)
\ldots\bbh_J^T(\tau)\epsilon_J(\tau)]^T=\lambda\bbepsilon(t-1)+
[\bbh_1^T(t)\epsilon_1(t)\ldots\bbh_J^T(t)\epsilon_J(t)]^T.
\end{equation}
Because under (a3) the zero-mean $\{\epsilon_j(t)\}_{j\in\calJ}$ are 
independent of 
$\bbz(t-1)$ [cf. \eqref{zstate_Ch4}], it follows readily that 
$\bbR_{z\bbepsilon}(t):=E[\bbz(t-1)\bbepsilon^{T}(t)]=\lambda E[\bbz(t-1)
\bbepsilon^{T}(t-1)]$. Plugging the expression for $\bbz(t-1)$ and carrying out
 the expectation yields
\begin{align}
E[\bbz(t-1)\bbepsilon^{T}(t-1)]&{}= 
\bbPsi E[\bbz(t-2)\bbepsilon^{T}(t-1)]+\bbPsi\left[\begin{array}{c} 
\bbR_{h,\lambda}^{-1}(\bbP_{\alpha}-\bbP_{\beta})\\
\bbC
\end{array}\right]E[\bbeta(t-3)\bbepsilon^{T}(t-1)]\nonumber\\
&+\bbPsi\left[\begin{array}{c} \bbR_{h,\lambda}^{-1}\\
\mathbf{0}_{Jp\times Jp}
\end{array}\right]E[\bar{\bbeta}(t-3)\bbepsilon^{T}(t-1)]
+\left[\begin{array}{c} \bbR_{h,\lambda}^{-1}\\
\mathbf{0}_{Jp\times Jp}
\end{array}\right]E[\bbepsilon(t-1)\bbepsilon^{T}(t-1)]\nonumber\\
&=\bbPsi\bbR_{z\bbepsilon}(t-1)+\left[\begin{array}{c} \bbR_{h,\lambda}^{-1}\\
\mathbf{0}_{Jp\times Jp}
\end{array}\right]\bbR_\bbepsilon(t-1).\label{derived_corr_rec}
\end{align}
The second equality follows from the fact that the zero-mean communication 
noise vectors are independent of $\bbepsilon(t-1)$.
Scaling \eqref{derived_corr_rec} by $\lambda$ yields the desired result.

% % % % % % % % % % % % % % % % % % % % % % % % % % % % % % % % % % % % % % % %
%                         Appendix E                                          %
% % % % % % % % % % % % % % % % % % % % % % % % % % % % % % % % % % % % % % % %

\noindent\normalsize \emph{\textbf{E. Structure of matrices $\bbP_\alpha$, 
$\bbP_\beta$, $\bbR_{\bar{\bbeta}}$, $\bbR_{\bbeta}$, 
$\bbR_{\bar{\bbeta}_{\lambda}}$, and $\bbR_{\bbeta_{\lambda}}$}}: In 
order to relate the noise supervectors $\bbeta_\alpha(t)$ 
and $\bbeta_\beta(t)$ with $\bbeta(t)$ in 
\eqref{noisevector_Ch4},
introduce two $Jp\times (\sum_{j=1}^J|\calN_j|)p$ matrices
$\bbP_{\alpha}:=[\bbp_1\ldots\bbp_J]^T$ and
$\bbP_{\beta}:=[\bbp'_1\ldots\bbp_J']^T$. The
$(\sum_{j=1}^J|\calN_j|)p\times p$ submatrices $\bbp_j$, $\bbp'_j$ are
given by $\bbp_j:=[(\bbp_{j,1})^T\ldots(\bbp_{j,J})^T]^{T}$ and
$\bbp_j':=[(\bbp_{j,1}')^T\ldots(\bbp_{j,J}')^T]^{T}$, with
$\bbp_{j,r},\bbp_{j',r}$ defined for $r=1,\ldots,J$ as
\begin{equation*}%\label{rjbrprimejbbvecs}
\bbp_{j,r}^T:=\left\{\begin{array}{cc}
\frac{c}{4}\bbb_{|\mathcal{N}_{r}|,r(j)}^T\otimes\bbI_p & \textrm{if
}
j\in\mathcal{N}_{r}\\
\mathbf{0}_{p\times |\mathcal{N}_{r}|p} & \textrm{if }
j\notin\mathcal{N}_{r}
\end{array}
\right.,\quad
 (\bbp_{j,r}')^T:=\left\{\begin{array}{cc}
\frac{c}{4}\mathbf{1}_{1\times |\mathcal{N}_{r}|}\otimes
\bbI_{p} & \textrm{ if } r=j\\
\mathbf{0}_{p\times |\mathcal{N}_{r}|p} & \textrm{if } r\neq j
\end{array}
\right..
\end{equation*}
Note that $r(j)\in\{1,\ldots,|\mathcal{N}_{r}|\}$ denotes the order in
which $\bbeta_j^{r}(t)$ appears in
$\{\bbeta_{j'}^{r}(t)\}_{j'\in\mathcal{N}_{r}}$ [cf. \eqref{noisevector_Ch4}].
It is straightforward to verify that $\bbeta_{\alpha}(t)=\bbP_\alpha\bbeta(t)$ 
and $\bbeta_{\beta}(t)=\bbP_\beta\bbeta(t)$.

Moving on to characterize the structure of $\bbR_{\bar{\bbeta}}$ and  
$\bbR_{\bbeta}$, from 
\eqref{baretajt_Ch4} and recalling
that communication noise vectors are assumed uncorrelated in space [cf. (a3)], 
it follows that
\begin{equation*}
\bbR_{\bar{\bbeta}}=\textrm{bdiag}\left(\sum_{j'\in\calN_1\backslash\{1\}}\bbR_
{\bbeta_{1,j'}},\ldots,
\sum_{j'\in\calN_J\backslash\{J\}}\bbR_{\bbeta_{J,j'}}\right).
\end{equation*}
Likewise, it follows from \eqref{noisevector_Ch4} that $\bbR_{\bbeta}$
is a block diagonal matrix with a total of $\sum_{j=1}^{J}|\calN_j|$
diagonal blocks of size $p\times p$, namely
\begin{equation*}
\bbR_{\bbeta}=\textrm{bdiag}\left(\{\bbR_{\bbeta_{j',1}}\}_{j'\in\calN_{1}},
\ldots,
\{\bbR_{\bbeta_{j',J}}\}_{j'\in\calN_{J}}\right).
\end{equation*}
Note also that the blocks $\bbR_{\bbeta_{j,j}}=\mathbf{0}_{p\times p}$ for
all $j\in\calJ$, since a sensor does not communicate with itself. In both 
cases, the block diagonal structure of the covariance matrices is due to the 
spatial uncorrelatedness of the noise vectors.

What is left to determine is the structure of $\bbR_{\bar{\bbeta}_{\lambda}}$ 
and $\bbR_{\bbeta_{\lambda}}$. From \eqref{noiselambdavectors} one readily 
obtains
\begin{equation}
\bbR_{\bar{\bbeta}_\mu}= \left[\begin{array}{c} \bbR_{h,\lambda}^{-1}\\
\mathbf{0}_{Jp\times Jp}
\end{array}\right]\bbR_{\bar{\bbeta}}\left[\begin{array}{c} 
\bbR_{h,\lambda}^{-1}\\
\mathbf{0}_{Jp\times Jp}
\end{array}\right]^{T},{\quad}\bbR_{\bbeta_\mu}=\left[\begin{array}{c} 
\bbR_{h,\lambda}^{-1}  
(\bbP_{\alpha}-\bbP_{\beta})\\
\bbC
\end{array}\right]\bbR_{\bbeta}\left[\begin{array}{c} \bbR_{h,\lambda}^{-1}
(\bbP_{\alpha}-\bbP_{\beta})\\
\bbC
\end{array}\right]^{T}.
\end{equation}
%

% % % % % % % % % % % % % % % % % % % % % % % % % % % % % % % % % % % % % % % %
%                         References                                          %
% % % % % % % % % % % % % % % % % % % % % % % % % % % % % % % % % % % % % % % %
% 
 
\newpage
\bibliographystyle{IEEEtranS}
\bibliography{IEEEabrv,biblio}

% % % % % % % % % % % % % % % % % % % % % % % % % % % % % % % % % % % % % % % %
%                         Figures                                             %
% % % % % % % % % % % % % % % % % % % % % % % % % % % % % % % % % % % % % % % %

\newpage

\begin{figure}[h]
\centering
\includegraphics[width=3in]{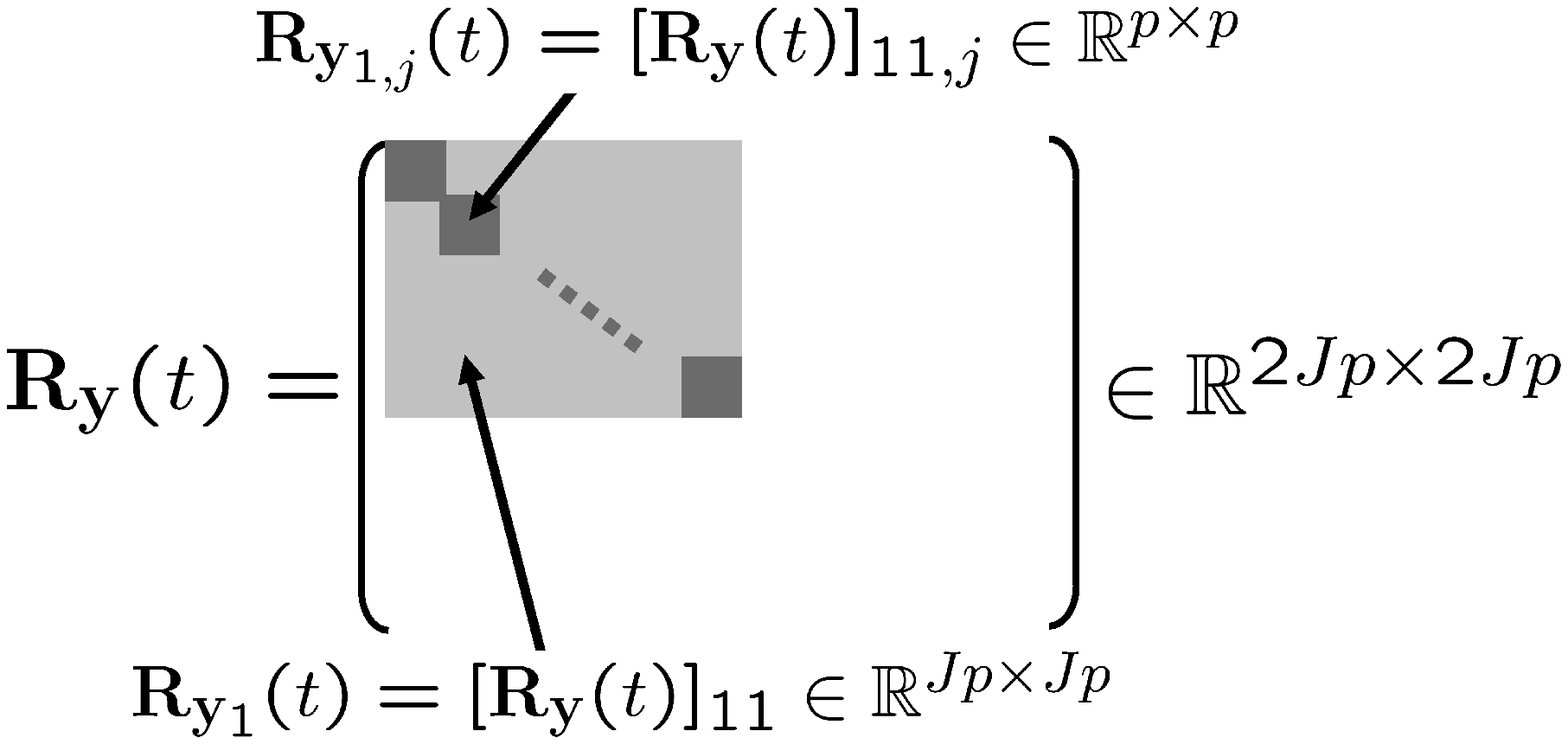}
\caption{The covariance matrix $\bbR_y(t)$ and some of its inner
submatrices that are relevant to the performance evaluation of the D-RLS algorithm.}
 \label{blocks}
\end{figure}

\begin{figure}[h]
\centering
\includegraphics[width=5in]{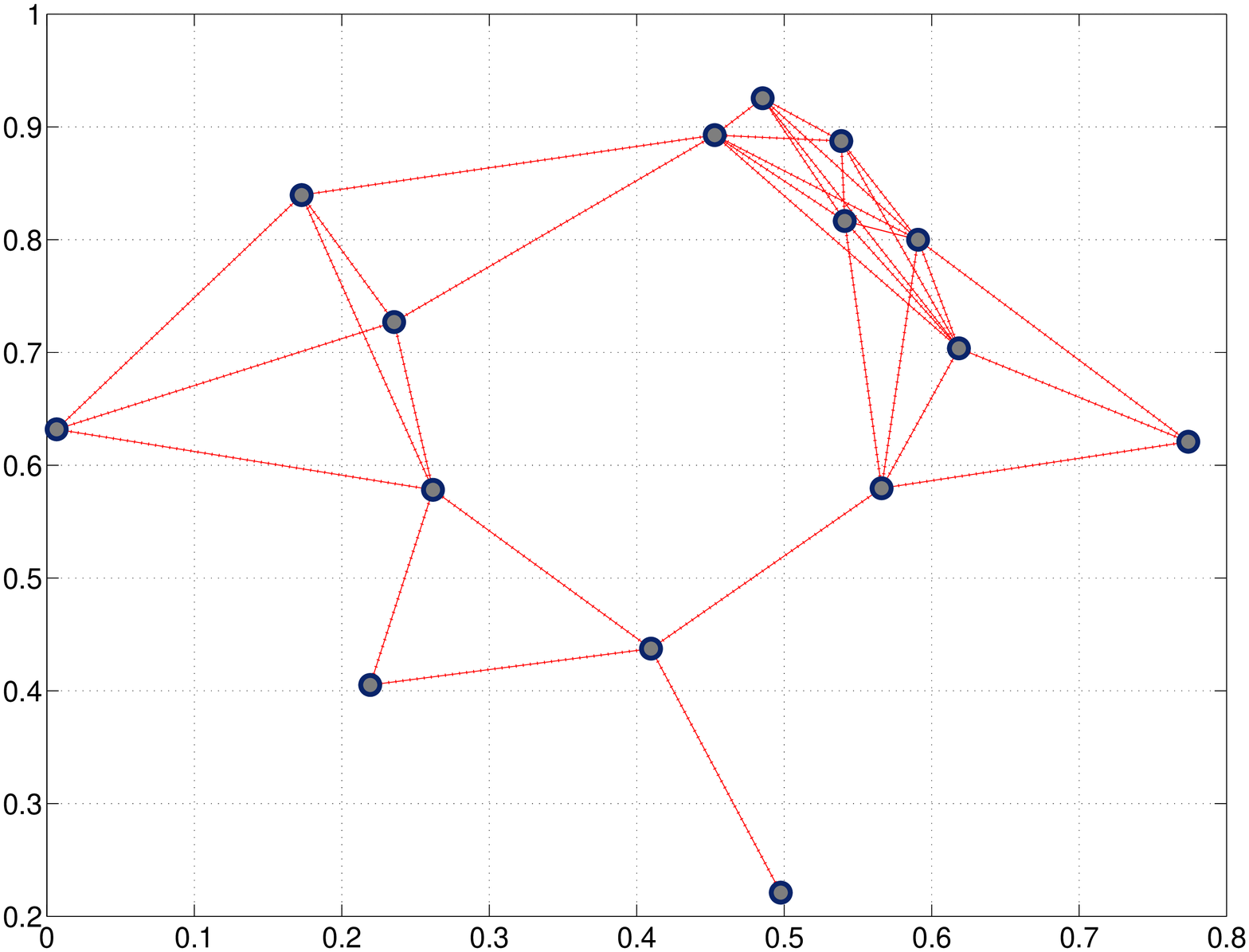}
\caption{An ad hoc WSN with $J=15$ sensors, generated as a realization of the
random geometric graph model on the unity square, with communication range $r=0.3$.}
 \label{wsn}
\end{figure}

\begin{figure}[h]
\centering
\includegraphics[width=5in]{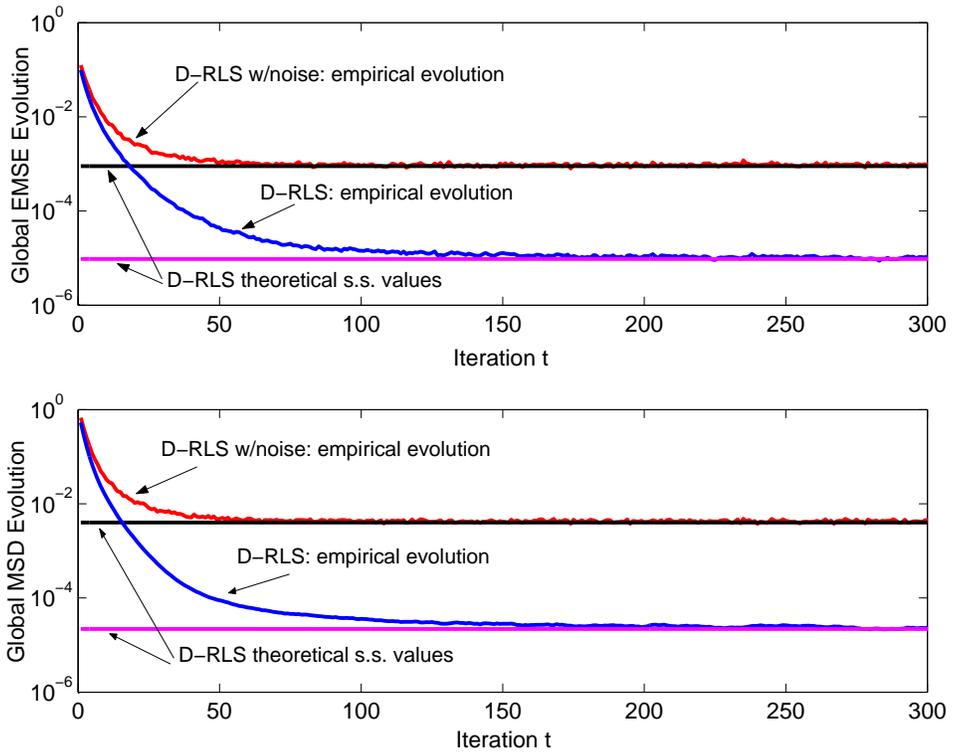}
\caption{Global steady-state performance evaluation. D-RLS is ran with ideal 
links
and when communication noise with variance $\sigma_\eta^2=10^{-1}$ is present.}
 \label{global_perf_DRLS}
\end{figure}

\begin{figure}[h]
\centering
\includegraphics[width=6in]{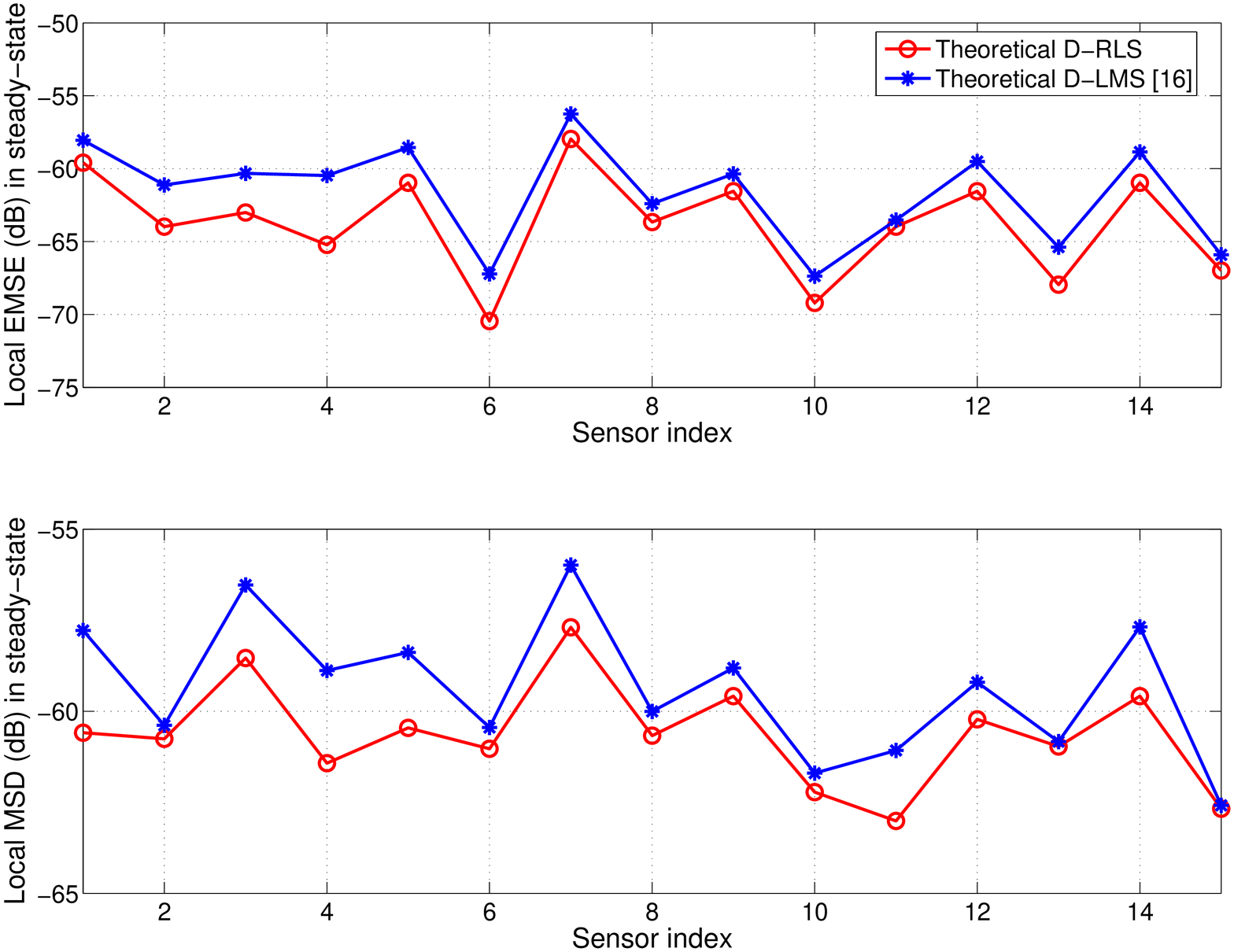}
\caption{Local steady-state performance evaluation. D-RLS is compared to the
D-LMS algorithm in~\cite{Gonzalo_Yannis_GG_DLMS_Perf}.} \label{local_perf_DRLS}
\end{figure}

%\begin{figure}[h]
%\begin{center}
%\centerline{\epsfig{file=wsn.eps,width=0.4\linewidth}} \caption{An ad hoc WSN with $J=30$ sensors.}
%\label{WSN}
%\end{center}\vspace{-1.8cm}
%\end{figure}
%
%\begin{figure}[h]
%\begin{center}
%\centerline{\epsfig{file=Fig1.eps,width=0.4\linewidth}\hspace*{0.3in}
%\epsfig{file=Fig2.eps,width=0.4\linewidth}} \caption{Global
%network performance in a distributed power spectrum estimation
%task: (left) MSE (learning curve); (right) MSD.} \label{MSE_MSD}
%\end{center}\vspace{-1.8cm}
%\end{figure}
%
%\begin{figure}[h]
%\begin{center}
%\centerline{\epsfig{file=Fig3.eps,width=0.4\linewidth}\hspace*{0.3in}
%\epsfig{file=Fig4.eps,width=0.4\linewidth}} \caption{(left) Local
%(per-sensor) performance in a distributed power spectrum estimation
%task; (right) Tracking with STD-RLS.} \label{local_tracking}
%\end{center}\vspace{-1.8cm}
%\end{figure}

\end{document}